\documentclass [a4paper, 11pt]{article}

\usepackage[margin=2.72cm]{geometry}
\usepackage{setspace}
\usepackage{amsmath}
\usepackage{amssymb}
\usepackage{amsthm}
\usepackage{gensymb}
\usepackage{bbm}
\usepackage{color}
\usepackage[dvipsnames]{xcolor}
\usepackage{graphicx}
\usepackage{caption}
\usepackage{subcaption}
\usepackage{mathrsfs}
\usepackage{multirow,array}
\usepackage{bbm}
\usepackage{xfrac}
\usepackage{geometry}

\usepackage{empheq}
\usepackage[most]{tcolorbox}
\usepackage[T1]{fontenc}
\usepackage[bookmarksnumbered=true]{hyperref} 
\hypersetup{
     colorlinks = true,
     linkcolor = red,
     anchorcolor = blue,
     citecolor = blue,
     filecolor = blue,
     urlcolor = blue
     }
\usepackage{natbib}

\newtheorem{definition} {Definition}
\newtheorem{proposition} {Proposition}
\newtheorem{lemma} {Lemma}
\newtheorem{lemmaA} {Lemma A}
\newtheorem{theorem} {Theorem}

\newtheorem*{claim*} {Claim}

\newtheorem{corollary}  {Corollary}

\begin{document}

\title{Pricing with algorithms\footnote{Lamba: Pennsylvania State University, \href{mailto: rlamba@psu.edu}{\texttt{rlamba@psu.edu}};  Zhuk: \href{mailto: sergey.zhuk@gmail.com}{\texttt{sergey.zhuk@gmail.com}}. We are grateful to Nageeb Ali, Yu Awaya, Tilman B\"{o}rgers, In-Koo Cho, Yuhta Ishii, Vijay Krishna, Erik Madsen, Yusufcan Masatlioglu, Andrew McClellan, David Miller, Scott Page, Andrew Sweeting, Darshana Sunoj and seminar participants at Pennsylvania State University, University of Maryland and NBER/CEME Decentralization Conference 2022 for their comments and suggestions.}}
\author{Rohit Lamba \hspace{30mm} Sergey Zhuk}
\date{June 2022}
\maketitle

\begin{abstract}
  This paper studies Markov perfect equilibria in a repeated duopoly model where sellers choose algorithms. An algorithm is a mapping from the competitor's price to own price. Once set, algorithms respond quickly. Customers arrive randomly and so do opportunities to revise the algorithm.  In the simple game with two possible prices, monopoly pricing is the unique equilibrium outcome for standard functional forms of the profit function. More generally, with multiple prices, market power is the rule--in all equilibria, the expected payoff of both sellers is above the competitive outcome, and that of at least one seller is eventually close to or above the monopoly outcome. Sustenance of such collusion seems outside the scope of standard antitrust laws for it does not involve any direct communication.

    \end{abstract}
\newpage

\section{Introduction}

There is growing concern amongst academics, policymakers, and increasingly consumers that algorithms are engendering otherwise outlawed collusive practices amongst firms in various industries. For example, writing in {\it Science} magazine, \citet*{science_AI_collusion} remark:

\begin{quote}\begin{small}
    Collusion is generally condemned by economists and policy-makers and is unlawful in almost all countries. But the increasing delegation of price-setting to algorithms (1) has the potential for opening a back door through which firms could collude lawfully (2). Such algorithmic collusion can occur when artificial intelligence (AI) algorithms learn to adopt collusive pricing rules without human intervention, oversight, or even knowledge.\end{small}
\end{quote}

\noindent There is growing empirical evidence for this too. Studying adoption of algorithmic pricing in German gas stations, \citet*{German_gasoline_AI} find 

\begin{quote} \begin{small}
    Adoption increases margins, but only for non-monopoly stations. In duopoly markets, margins increase only if both stations adopt, suggesting that AP has a significant effect on competition.\end{small}
\end{quote}

\noindent Finally, new experimental work also suggests that algorithms may lead to collusion using adaptive strategies oft invoked in dynamic games. \citet*{experiment_AI_collusion} write: 

\begin{quote}\begin{small}
    The high prices are sustained by collusive strategies with a finite phase of punishment followed by a gradual return to cooperation. This finding is robust to asymmetries in cost or demand, changes in the number of players, and various forms of uncertainty.\end{small}
\end{quote}

The objective of this paper is to offer a parsimonious model that captures the essence of why algorithms may be abetting collusion.  We focus on the two specific properties of algorithms (also outlined in \citet{competiton_algo}). First, by design they allow a quick response time to price changes by competitors. Second, they offer temporary commitment to a pricing strategy---once an algorithm is chosen, price setting is delegated to it, till the algorithm is updated. The paper argues that these properties indeed facilitate collusion--- the market {\it naturally} converges to prices higher than the competitive outcome. Moreover, these choices looks legal on the books since they does not require any direct communication.




More specifically, we study a duopoly model of price setting, where the strategy for each seller is to choose an algorithm. An algorithm is simply a mapping from competitor's price to own price. Once the two algorithms are chosen, they automatically set prices in response to each other. Customers arrive according to some Poisson process. Upon arrival, a probabilistic demand function determines which seller the customer may buy from---a lower price leads to a higher likelihood of being chosen. In addition, each seller gets an opportunity to revise their algorithm according to a separate Poisson process. The dynamic game essentially follows a sequence of customer arrivals and opportunities to revise algorithms (see Section \ref{section model} for the model). 

The key modeling assumption is that ``$\Delta t$", i.e. the unit that measures the reaction time of algorithms, is small. It delivers three mechanisms: First, sellers can experiment and learn the opponent's algorithm very quickly, the learning cost is at most linear in the set of prices and the unit of time. So effectively when they choose their algorithm, they are best responding to the algorithm of the opponent. Second, prices converge quickly for a fixed set of algorithms and the convergent prices are what essentially matter for accounting payoffs. Third, since the time period of price changes is smaller than the average time sellers have to change algorithms, this offers them temporary commitment to a pricing strategy, in contrast to a standard repeated game. Once two algorithms are set, they react mechanically to the opponent's price, till one of the sellers gets a new opportunity revise its algorithm.

In fact, the first step in the analysis is to show that this dynamic game reduces to an asynchronous repeated game, that is, a repeated game in which players rotate in picking strategies. We then focus on Markov perfect equilibria of this game. The main result of the paper is that some level of collusion is inevitable, i.e. it is always attained, for all specifications of the model, and that these collusive equilibria can be further characterized. For starters, the payoff of both sellers is greater than the competitive outcome.\footnote{By ``competitive outcome", we mean here the Nash equilibrium of the one-shot game.} Moreover, exercise of market power is omnipresent---in all equilibria, the payoff of at least one seller is bounded from  below by a constant that approximates the monopoly profit.\footnote{In a duopoly model, we abuse terminology slightly and term monopoly profit to be half of the profit that would accrue to a single seller who sets prices to maximize its profit.} Further still, for certain specifications, repeated play of the monopoly strategy is the unique convergent outcome. In short, collusion is inevitable and extreme collusion is commonplace. 

The model is closely related to \citet{bruno_tacit}, which presents a remarkably elegant way to capture the idea of collusion through algorithms. \citet{bruno_tacit}, however, characterizes subgame perfect equilibria and asserts the main result for the limit case of zero effective discounting, that is zero revision opportunities for the sellers, which we argue is unrealistic. Our tack on the problem delivers results for all values of effective discounting, and in addition presents a much more exhaustive characterization of (Markov) equilibrium outcomes. A detailed discussion is presented in Section \ref{sec:comparison_bruno}. 


Our analysis starts with a special case where the sellers have access to two prices--- the competitive one and the pure monopoly one (Section \ref{sec:PD}). Within that realm, we consider a general 2$\times$2 matrix game that satisfies standard restrictions for the Prisoner's Dilemma game. It is shown that for two standard specifications--when the profit function is either quadratic or takes the discrete choice form used in various applied models-- repeated play of the monopoly price is the unique equilibrium outcome. Further, for general profit functions too, three types of equilibria emerge, two of which lead to monopoly as the unique outcome and the payoffs in the third too are bounded away from the competitive outcome. Recollect that repetition of static Nash prices is always an equilibrium in the standard repeated games model, not here. 

Here is a brief intuition for the inevitability of collusion in the 2$\times$2 model: The universe of all algorithms can be reduced to four types. A seller could always chooses the competitive price, no matter what the opponent does, call this algorithm/strategy $s_C$; the seller can also always choose the monopoly price, call this $s_M$. In addition, the seller can play tit-for-tat, that is, choose whatever price the opponent chooses, call this $s_T$ or it can play reverse-tit-for-tat, which means picking the opposite of what the opponent chooses, call this $s_R$. 

Suppose seller $A$ chooses $s_C$, that is, always pick the competitive price. Then, upon their turn to set an algorithm, seller $B$ can also pick $s_C$ so that the competitive price is selected by both. In a standard repeated game this would result in repetition of static Nash, which is of course an equilibrium outcome. But in this set up, is it optimal for $B$ to choose $s_C$? It is not: if instead $B$ chose $s_T$, then the competitive outcome will be played till $A$ has an opportunity to update their algorithm from $s_C$. However, once $A$ has a chance to revise their algorithm, the best response is also to choose $s_T$ and start off the new phase with the monopoly price. Since both are playing $s_T$, the monopoly price is selected repeatedly, which generates a higher payoff. This points towards why repetition of the competitive price cannot be an equilibrium outcome. The repeated choice of $s_C$ does not satisfy the perfectness criterion in Markov perfect equilibrium.

Building on this logic further by comparing 4x4 matches of all algorithms, one can in fact prove that the repeatedly play of the monopoly price emerges globally (i.e. for all parameters) as an equilibrium outcome. Moreover, for many reasonable specifications of payoffs, it is the unique outcome, hence the claim of extreme collusion being commonplace. 
 
We then take these insights to a general model with any finite  grid of prices (Section \ref{section multiple}). This of course complicates the analysis (at least theoretically) by orders of magnitude. Even with three possible prices there are 27 algorithms to choose from and with four there are 256 and so on. Despite this complexity we draw a general theoretical lesson. It says that payoffs in all Markov perfect equilibrium are such that: (i) both sellers receive a payoff higher than the competitive outcome, and (ii) at least one sellers' payoff is bounded from below by a constant that approximates monopoly profit. 

At a high level, the intuition for this result is similar to the 2x2 model. The low price equilibria are ``learnt away" by following a chain of temporary commitments to a higher price. Within each commitment phase of set algorithms, prices converge quickly because of the smallness of response times, so only the convergent outcomes are taken to be payoff relevant. And while setting algorithms, sellers take into account future opportunities for the competitor to revise their algorithms and so in the larger repeated game, algorithms too eventually settle on to a stable pair or cycle of pairs. 

In the 2x2 model, a two step chain took us from one seller choosing $s_C$ to both eventually choosing $s_T$, which sustained the monopoly outcome. With multiple possible prices, more  complex chains can be formed. We consider a subset of these to show that in any equilibrium at least one seller will end up guaranteeing themselves a payoff which is bounded from below by a high enough constant. 

To prove the main results, we make simplifying assumptions on how payoffs are calculated and how cycles (as opposed to convergent prices) are evaluated for a fixed set of algorithms. These are relaxed in extensions (Section \ref{sec:extensions}) to show that the basic insights carry through. We also extend equilibrium notion to sugbame perfect and show that a much larger class of equilibria emerge where both collusion and conflict can be sustained. Moreover, the specific sequence of algorithms required to sustain these payoffs are quite complex. 

In terms of techniques, the paper draws on the literature on  asynchronous repeated games, especially \citet{maskin_tirole1, maskin_tirole2}. These set of papers explore temporary commitment by fixing a player's strategy for some finite amount of time and rotating the opportunity to revise them, hence the term `asynchronous' (see also \citet{asyn_rg}). The technical analysis is also related to repeated games with renegotiation (see for example \citet{david_renego}, \citet{miller_watson_renego}, \citet{contest_norms}) in that opportunities to set algorithms can be seen as renegotiation constraints. In addition, algorithms carry the flavor of finite automaton in repeated games, as for example in \citet{finite_automaton}. Most of this literature however has players setting their automaton at the beginning with no recourse for revisions. Finally, in allowing for random opportunities to revise strategies, we also relate to the literature on revision games, see for example \citet{revision_games}. 

We of course speak to a now burgeoning literature on the foundations and implications of algorithmic pricing. \citet{competiton_algo} document and then model asymmetries in price dispersion arising out of algorithmic pricing. \citet*{algo_pakes} study the impact of different learning protocols used by artificial intelligence on algorithmic pricing. \citet*{platform_algo} explore experimentally the ability of a platform to increase consumer welfare by reducing prices in differentiated product models where sellers use learning algorithms. \citet{harrington_algo} presents a model in which sellers outsource pricing to a third party that uses an algorithm and finds that the algorithm does not reduce competition but it makes prices more sensitive to underlying variation in demand. \citet*{mallesh_bhai_ka_paper} argue that when competing sellers are using algorithms to learn some features of demand, collusion emerges because the sellers end up running correlated experiments. 

The study of collusion (or more broadly cooperation) has been a central question in much of game theory (as far back as \citet{aumann_folk} and \citet{old_coop}). \citet{green_porter} argued that simply observing the competitors' prices, as opposed to entire strategies, is sufficient to sustain collusion. More recently, \citet{awaya_krishna} show that communication can in fact help sustain collusion even if firms may not be able to track others' prices. \citet{sugaya_folk} proves a meta folk theorem for repeated games with private monitoring. In comparison to these, we keep the structure of information simple in that all actions are observable, but emphasize that quick reaction and temporary commitment offered by algorithms makes collusion not just possible; in fact some level of collusion is inevitable and extreme collusion is  commonplace.\footnote{There is also a significant applied and empirical literature on the economics of detecting and understanding mechanics of collusion, especially in auctions. See \citet{asker_collusion} and \citet*{chassang_collusion} for some recent work, and \citet{mm_collusion_book} for a textbook treatment. Collusion due to the use of algorithms though is a relatively new addition to this cannon. In recent work, \citet{AI-auctions} explore implications of simple learning algorithms for the sustenance of collusion in auctions.} 

While the model is stylized, it captures some essence of why automated price settings can lead to sustained markups that regulators can no longer ignore. Indeed, the endeavor here is to tease out sustained monopoly as the unique (in the two-price model) outcome or at least extreme collusion as the frequent (in the multi-price model) outcome, so that the theoretical point can be made starkly.

The challenge for anti-trust authorities assumes great significance for old methods of defining, detecting and hence describing collusion in a court of law are limitedly effective. If collusion can be sustained without explicit communication, how do we proceed? In fact, \citet{harrington_ai} argues that ``collusion by software programs which choose pricing rules without any human intervention is not a violation" of standard anti-trust laws used in the United States. This is a puzzle that several antitrust and competition commissions across the world are now encountering or will do in the near future, which is an invitation for many more theoretical and empirical studies to inform law and policy. 

\section{Model}

\label{section model}

\subsection{Market and prices\label{subsec:market}}

Two sellers $i=A,B$ are selling a product. Each sets a price $p^{A}$ and $p^{B}$. The marginal costs are normalized to $0$. Customers arrive at a Poisson rate $\lambda$, so the probability of having a customer in the interval $\Delta t$ is approximately $\lambda\cdot\Delta t$. Each customer purchases at most one unit of the product. 

Demand is defined probabilistically: At current prices $p$ and $q$, the probability that a customer purchases the product from the firm with price $p$ is given by $d(p,q)$. So, total demand is given by $d(p,q) + d(q,p)$. Fix the payoff function to be 
\begin{equation}
\pi(p,q)=p\cdot d(p,q).
\end{equation}
Then, the expected payoff from one customer for each seller is given by 
\begin{equation}
\begin{array}{c}
\pi^{A}=\pi(p^{A},p^{B})\\
\pi^{B}=\pi(p^{B},p^{A})
\end{array}
\end{equation}
We will impose the following assumptions om the payoff function $\pi(p,q)$: \\ \vspace{-3mm}

\noindent {\bf Assumption 0}. (i) $\pi(p,q)$ is increasing in $q$ and for each $q$ there is a unique global maximum of $\pi(p,q)$ in $p$ which is increasing $q$; (ii) there is a unique static Nash equilibrium with prices $(p_{C},p_{C}$) (where $C$ stands for ``competitive'' price); (iii) the joint profit $\pi(p,q)+\pi(q,p)$ is maximized at $p=q=p_{M}>p_{C}$ (where $M$ stands for ``monopoly'' price); (iv) for any $p\geqslant \text{arg}\max_{p'} \pi(p', p_M)$, the point $(\pi(p,p_M), \pi(p_M,p))$ is on the boundary of the convex hull of the payoffs set $\{(\pi(p,q), \pi(q,p))\ for\  p,q\in[p_C,p_M]\}$.   \\ \vspace{-3mm}


Parts (i) to (iii) are standard assumptions that correspondent $\pi$ to a standard Bertrand duopoly. Part (iv) is a useful technical assumption that helps us narrow down the number of deviations to consider when stating the general result (Theorem \ref{prop:multiple_prices}). Several commonly used market demand functions (including the linear demand and discrete choice demand mentioned in Corollary \ref{cor:PD_monopoly_condition}) satisfy these assumptions. 

We also restrict the set of possible prices for each seller to a discrete subset of equally spaced numbers between the competitive $p_{C}$ and monopoly prices $p_{M}$: 
\begin{equation}
p_{C}=p_{1}<\dots<p_{K}=p_{M}.
\end{equation}
Fix $\Delta=p_{i}-p_{i-1}$ for $i=1,\dots,K$ to be the difference between two consecutive prices. In a slight abuse of notation,
we denote:
\begin{equation}
\begin{array}{l}
p_{M-\Delta}=p_{M}-\Delta=p_{K-1}\\
p_{M-2\Delta}=p_{M}-2\cdot\Delta=p_{K-2}\\
\dots\\
p_{C+\Delta}=p_{C}+\Delta=p_{1}
\end{array}
\end{equation}

\subsection{Algorithms}

\label{algos}

An algorithm, in its simplest form, is a function $s:$ $p^{-i}\rightarrow p^{i}$ that determines which price $p^{i}$ is chosen by seller $i$ in response to the competitor's price  $p^{-i}$.\footnote{Algorithms broadly are a set of instructions that mechanically map some input information to an pre-specified set of outputs. Different notions of algorithms have been adopted thus far in the economics literature. The definition used here mostly follows \citet{competiton_algo}. In \citet{algo_pakes} and \citet{platform_algo}, sellers adopt the Q-learning algorithm, often used in practice. In \citet{harrington_algo}, an algorithm maps demand shocks not observable to the sellers to a prescribed price. A generalized version of our notion, also used in \citet{bruno_tacit}, defines an algorithm to be a mapping $s:\omega^{i},p^{-i}\rightarrow\nu^{i},p^{i}$, where $\omega^{i}$ represents the current state of the algorithm of seller $i$ and $\nu^{i}$ is the subsequent state. As discussed later in Section \ref{final remarks}, if the cardinatlity of these states is not too large, our results can largely be extended to such algorithms.} Prices are observable but sales of the competitor are not, or equivalently algorithms react only to prevailing prices, which we think is a reasonable assumption, especially for online sellers.

We assume algorithms can change prices once per small interval
of time $\Delta t$. Suppose, the prices react at the following calendar times:
$0,\Delta t,\ 2\cdot\Delta t,\ 3\cdot\Delta t,\dots$ If $s^{A}$
and $s^{B}$ are the corresponding algorithms of the two sellers, then
the prices will evolve according to
\begin{equation}
\begin{array}{c}
p_{t+\Delta t}^{A}=s^{A}\left(p_{t}^{B}\right)\\
p_{t+\Delta t}^{B}=s^{B}\left(p_{t}^{A}\right)
\end{array}\label{eq:prices}
\end{equation}

If the algorithms of the two sellers are temporarily fixed, then there
are only two possible outcome for the price sequences generated from (\ref{eq:prices}).
\begin{enumerate}
\item They converge to fixed price pair ($p^{A},p^{B})$ such that:
\begin{equation}
p^{A}=s^{A}(p^{B})\qquad p^{B}=s^{B}(p^{A})\label{eq:final_pair}
\end{equation}
\item They converge to a cycle of length $T$:
\begin{equation}
\begin{array}{ccccc}
p_{1}^{A}\rightarrow & p_{2}^{A}\rightarrow & \dots\rightarrow & p_{T}^{A}\rightarrow & p_{1}^{A}\rightarrow\dots\\
p_{1}^{B}\rightarrow & p_{2}^{B}\rightarrow & \dots\rightarrow & p_{T}^{B}\rightarrow & p_{1}^{B}\rightarrow\dots
\end{array}\label{eq:cycle}
\end{equation}
\end{enumerate}

The sellers can also periodically adjust their algorithms. The possibility to adjust the algorithms arrives at an exogenous Poisson process with intensity $\mu$; so the probability of the event happening in the interval $\Delta t$ is approximately $\ensuremath{\mu}\ensuremath{\cdot}\ensuremath{\Delta t}$. In addition, whenever one of the sellers adjusts their algorithm, they can pick an initial short sequences of prices. So, for a few subsequent intervals of length $\Delta t$, prices can be chosen by the seller intentionally, after which the price setting is delegated to the new algorithm.

This initial set of prices can be employed for two objectives. First, sellers experiment and learn the opponent's algorithm before setting their own. It takes at most $K$ prices adjustments ($K$ is the cardinality of price set) to fully understand the opponent's algorithm. The cost of this is proportional to $K\cdot\Delta t$, which is linear in $K$ and unit of time, and hence small for a small $\Delta t$. Second, upon deciphering the opponent's algorithm, the seller adjusting the algorithm can choose initial few prices to select their most preferred price sequence that satisfies (\ref{eq:prices}) for the prevailing algorithms $s^{A}$ and $s^{B}$. Think of this as the initial condition for the new system that is then delegated to determine prices. Assumption 1 below ensures that the initial prices for both these objectives, while important in setting algorithms, are not payoff relevant.\footnote{A different interpretation of similar modeling assumptions is used in \citet{bruno_tacit} (and to some extent in \citet{competiton_algo} as well). There the time between revision opportunities is used by the seller to decipher the opponent's algorithm. So, learning is captured exogenously through the random arrival times to set algorithms. The same interpretation can be applied in our setting. The experimentation interpretation, which we slightly prefer, is facilitated by the smallness of time required for price adjustments in our model. In Appendix E of the paper, we discuss the issue of experimentation further.}

\subsection{Simplifying assumptions \label{subsec:assumptions}}

In what follows, we make three simplifying assumptions. These make the overall analysis tractable and allow us to focus on the ``big picture'' of how quickness of response and temporary commitment offered by algorithms leads to collusion. In section \ref{sec:extensions} we relax two of these and show that our key results continue to hold.
\begin{itemize}
\item \textbf{Assumption 1. }The interval $\Delta t$ is sufficiently small (in particular, we need $r\Delta t\ll 1$ and $\mu\Delta t \ll 1$) so that initial price adjustments governed by (\ref{eq:prices}) can be ignored and only convergent prices, given by (\ref{eq:final_pair}) or \eqref{eq:cycle} are payoff relevant.
\item \textbf{Assumption 2. }Cycles are ruled out. A seller changing their
algorithms is not allowed to choose a new algorithm (and an initial price sequence)
such that together with the existing algorithm of the opponent the pair
results in a cycle (\ref{eq:cycle}).
\item \textbf{Assumption 3. }If a seller adjusting their algorithm is indifferent
between two options (in terms of continuation payoffs), the seller
would choose the one that is preferred by their opponent.
\end{itemize}

\textbf{Assumption 1} represents the two key features of algorithms that we attempt to model, that they respond quickly, and they offer temporary commitment to sellers. These are ensured by $r\Delta t\ll1$ and  $\frac1{\mu}\gg\Delta t$. As a consequence, only the convergent prices matter for payoffs. We view the smallness of $\Delta t$ as an important aspect of algorithmic pricing and thus a critical modeling choice for our analysis. The assumption that adjustments driven by (\ref{eq:prices}) pre-convergence can be ignored is made for simplicity and will be relaxed in Section \ref{sec:extensions} to show that all the basic insights carry through.

Our motivation for \textbf{assumption 2} is that you normally would not see cycles in reality. A cycle with frequently changing prices (for
example, every minute or even every hour) would be confusing for
consumers. Moreover, when consumers see frequently changing prices
they can time their purchases and wait for the best prices from each
of the sellers, which is not in the sellers' interest either.\footnote{Formally, disallowing cycles amounts to restricting the strategy space of the underlying stage game, that is, for a fixed strategy of the opponent, a restricted set of strategies is available to a seller. Since the probability of both sellers receiving opportunities to revise the algorithm at the same time is zero, \textbf{Assumption 2} still leads to a well defined dynamic game. One can rationalize this through a high cost of entering a cycle.} In Section \ref{sec:extensions} we show
that our results for the case of two possible prices 
do not change if we do indeed allow cycles.


\textbf{Assumption 3} is typically never binding. If 
two different algorithms generate different subsequent paths in terms
of prices, then the continuation payoffs would only be identical for
certain very special discount factors $\beta$ or for certain very
special payoff matrices $\{\pi(p,q)|p,q=p_{C},p_{C+\Delta},\dots,p_{M}\}$.
We make this assumption so that even in such special cases we do not have to deals with the multiplicity of optimal responses.

\subsection{Expected payoffs and reduction to an asynchronous repeated game}

Suppose that the current algorithm of seller $i$ is $s^{i}$ and
the current algorithm of their competitor is $s^{-i}$ and seller
$i$ was the last to adjust their algorithm. Denote by
\begin{itemize}
\item $\xi(s^{-i},s^{i})=(p^{-i},p^{i})$: the price pair to which algorithms would converge; it would be a price pair consistent with these two algorithms and since $i$ was the last to adjust their algorithm, amongst all such price pairs it will be the one most preferred by $i$.
\item $\pi(s^{i},s^{-i})=\pi(p^{i},p^{-i})$: the expected payoff from one customer for the seller who was the last to adjust their algorithm after we reach the convergent price pair. 
\item $\bar{\pi}(s^{-i},s^{i})=\pi(p_{-i},p_{i})$: the expected payoff from one customer for the seller who was {\it not} the last to adjust their algorithm after we reach the convergent price pair. 
\end{itemize}

Index all the moments when one of the sellers is adjusting
their algorithm with $k=0,1,2,\dots$. Suppose the seller adjusting
their algorithm at moment $k$ chooses the algorithm $s_{k}$, thus
we end up with the following sequence of algorithms $s_{0},s_{1},s_{2},\dots$. In this sequence all ``even'' algorithms correspond to one seller and all ``odd'' algorithms correspond to the other.\footnote{Suppose seller $A$ gets the first random draw to set their algorithm. Then even if they get a second random draw to adjust their algorithm before seller $B$'s first turn, it is never optimal to deviate from the current algorithm.}

Now consider some arbitrary adjustment moment $k$. Denote by $U_{k}$
the expected continuation payoff of the seller adjusting their algorithm,
and by $V_{k}$ the expected continuation payoff of the competitor.
With Assumption 1, we can calculate the continuation payoffs
recursively:

{\footnotesize{}
\begin{equation}
U_{k}=\int_{0}^{\infty}e^{-rt}\cdot e^{-\mu t}\left[\lambda\cdot dt\cdot\pi(s_{k},s_{k-1})+\mu\cdot dt\cdot V_{k+1}\right]=\frac{\lambda}{r+\mu}\pi(s_{k},s_{k-1})+\frac{\mu}{r+\mu}V_{k+1}\label{eq:U}
\end{equation}
\begin{equation}
V_{k}=\int_{0}^{\infty}e^{-rt}\cdot e^{-\mu t}\left[\lambda\cdot dt\cdot\bar{\pi}(s_{k-1},s_{k})+\mu\cdot dt\cdot U_{k+1}\right]=\frac{\lambda}{r+\mu}\bar{\pi}(s_{k-1},s_{k})+\frac{\mu}{r+\mu}U_{k+1}\label{eq:V}
\end{equation}
}

\noindent If we normalize the continuation payoffs as \begin{equation}
u_{k}=\frac{U_{k}}{\lambda/r} \text{ and } 
v_{k}=\frac{V_{k}}{\lambda/r},
\end{equation}
and denote normalized discount factor as
\begin{equation}
\beta=\frac{\mu}{r+\mu}\label{eq:beta}
\end{equation}
then we get the following two simple equations describing the  evolution of payoffs: 
\begin{equation}
\begin{array}{l}
u_{k}=(1-\beta)\cdot\pi(s_{k},s_{k-1})+\beta\cdot v_{k+1}\\
v_{k}=(1-\beta)\cdot\bar{\pi}(s_{k-1},s_{k})+\beta\cdot u_{k+1}
\end{array}\label{eq:uv}
\end{equation}
Thus, we have essentially converted our more complicated `continuous time
game' into an {\it asynchronous repeated game} in discrete time with discount factor $\beta$,
in which one seller adjusts their action (which is an algorithm) in
``even'' periods and the other in ``odd'' periods and the payoffs
in each in period are determined by finite matrices $\pi(s_{k},s_{k-1})$
and $\bar\pi(s_{k-1},s_{k})$.

A comment on the normalization above is in order. Note that the normalized expected payoffs from equations \eqref{eq:uv} do not directly depend on customer arrival rate $\lambda$. This happens because the rate is essentially a scaling factor that multiplies profits of each seller, and, thus does not affect their optimal behavior. The effective discount fact $\beta$ is determined by the ratio between the frequency of algorithm revision opportunities $\mu$ and the interest rate $r$ faced by the sellers. In realistic scenarios $r\ll\mu$, and, thus, $\beta$ is typically close to 1, but still bounced away from it. For example, if the interest rate is 5\% per year and each seller has (on average) 5 opportunities to revise their algorithm in a year, then $\beta=\frac{5}{0.05+5}\approx 0.99$.

\subsection{Markov perfect equilibria}

In our analysis we focus mainly on Markov Perfect Equilibria. In
such equilibria when one of the sellers is adjusting their algorithm, their best response depends only on the current algorithms used by their opponent. This means that if the seller would face the same algorithms of the opponent in the future, they would respond in the same manner. In Section \ref{sec:extensions} we discuss the more general class of subgame perfect equilibria.

Two quick things to note in motivating MPE: First, while equilibria with richer intermediate levels of collusion are possible with subgame as opposed to Markov perfection, they are typically rather complicated. They require each seller to have very specific beliefs about the behavior of their competitor, which may not hold muster in reality. Second, we want to point out that since the space of possible algorithms is typically quite large, even Markov Perfect Equilibria are not necessarily very simple. The reader will shortly learn these equilibira still cover a wide range of possibilities.\footnote{\citet{maskin_tirole2} (in Section 10 of the paper) provide  four reasons for studying MPE as a benchmark of sorts in repeated oligopoly models, and we refer the reader to their persuasive pitch for our choice of MPE. To be clear, this is not to say that characterizing SPE is not important, just that MPE is a good benchmark to work with. Moreover, it is perhaps more surprising that extreme collusion is widespread in MPE, that varying levels of collusion can be sustained in SPE for high discounting is to be expected.}

In a Markov Perfect Equilibrium, the current algorithm of the competitor, say $s$, is regarded as the state variable. Then, for each seller ($i=A,B)$ we can define:
\begin{itemize}
\item $f^{i}(s)$: best response of seller $i$ to the strategy $s$ of the opponent when $i$ is adjusting their algorithm.
\item $u^{i}(s)$: the continuation payoff of $i$ if $i$ is changing their algorithms and the current algorithms of the opponent is $s$.
\item $v^{-i}(s)$: the continuation payoff of $-i$ if $i$ is changing
their algorithm and the current algorithm of $-i$ is $s$.
\end{itemize}
In equilibrium the best response $f^{i}(s)$ and the continuation
payoffs $u^{i}$ and $v^{-i}$ should satisfy:
\begin{equation}
\begin{array}{l}
u^{i}(s)=\max_{s'}(1-\beta)\cdot\pi(s',s)+\beta\cdot v^{i}(s')\\
\ \\
f^{i}(s)\in\arg\max_{s'}(1-\beta)\cdot\pi(s',s)+\beta\cdot v^{i}(s')\\
\ \\
v^{-i}(s)=(1-\beta)\cdot\bar{\pi}(s,f^{i}(s))+\beta\cdot u^{-i}(f^{i}(s))
\end{array}\label{eq:MPE}
\end{equation}
The sequence of algorithms $s_{0},s_{1},s_{2},\dots$ in which each
algorithm is a best response to the previous one (i.e. for
any $k$, $s_{2k+1}=f^{i}(s_{2k})$ and $s_{2k+2}=f^{-i}(s_{2k+1})$)
will be termed an {\it equilibrium sequence} starting with $s_{0}$.

\section{Two prices and reduction to a Prisoner's dilemma game \label{sec:PD}}

\subsection{Expected payoff from a single customer}

We start our analysis with a simple case in which each seller $i=A,B$
can choose one of the two prices: $p_{M}$ (monopoly) or $p_{C}$ (competitive).
The expected payoffs from one customer can be described with a simple
two-by-two matrix (see Table \ref{tab:PD_payoffs}). Assumption 0 then reduces to the following conditions.

The payoff matrix has the Prisoners' dilemma type structure, that
is
\begin{equation}
\pi(p_{M},p_{C})<\pi(p_{C},p_{C})<\pi(p_{M},p_{M})<\pi(p_{C},p_{M}),\label{eq:PD_conditions}
\end{equation}
In addition, the monopoly outcome $(p_{M},p_{M})$
is superior from a utilitarian perspective
\begin{equation}
2\cdot\pi(p_{M},p_{M})>\pi(p_{C},p_{M})+\pi(p_{M},p_{C}).\label{eq:PD_so}
\end{equation}

In order to describe how possible equilibria depend on the payoff
matrix, it would be convenient to define
\begin{equation}
x=\frac{\pi(p_{C},p_{M})-\pi(p_{M},p_{M})}{\pi(p_{M},p_{M})-\pi(p_{C},p_{C})}\qquad y=\frac{\pi(p_{C},p_{C})-\pi(p_{M},p_{C})}{\pi(p_{M},p_{M})-\pi(p_{C},p_{C})}.\label{eq:PD_xy}
\end{equation}
Then, the Prisoners' dilemma game in Table \ref{tab:PD_equivalent}
is ``equivalent" to the original game in the sense that set of equilibria in both Tables \ref{tab:PD_payoffs} and \ref{tab:PD_equivalent} are exactly the same. So, we can describe all possible equilibria in the static game in terms of $x$ and $y$, and in the dynamic game in terms of $x, y$ and $\beta$.

The parameter $x$ shows by how
much the deviation from monopoly, $\pi(p_{C},p_{M})$, exceeds the monopoly
payoff, $\pi(p_{M},p_{M})$, and $y$ shows by how much $\pi(p_{M},p_{C}$)
is below the competitive payoff, $\pi(p_{C},p_{C})$. Both are measured relative to the total distance $\pi(p_{M},p_{M})-\pi(p_{C},p_{C})$. Conditions
\eqref{eq:PD_conditions} and \eqref{eq:PD_so} are now equivalent to
\begin{equation}
\begin{array}{l}
x,y>0\\
y>x-1
\end{array}
\end{equation}

\begin{table}
\centering{}%
\begin{tabular}{|c|c|c|}
\hline 
$A\backslash B$ & $p_{M}$ & $p_{C}$\tabularnewline
\hline 
$p_{M}$ & $\pi(p_{M},p_{M}),\,\pi(p_{M},p_{M})$ & $\pi(p_{M},p_{C}),\,\pi(p_{C},p_{M})$\tabularnewline
\hline 
$p_{C}$ & $\pi(p_{C},p_{M}),\,\pi(p_{M},p_{C})$ & $\pi(p_{C},p_{C}),\,\pi(p_{C},p_{C})$\tabularnewline
\hline 
\end{tabular}\caption{\label{tab:PD_payoffs}Expected payoffs from one customer for the two prices model}
\end{table}
\begin{table}
\centering{}%
\begin{tabular}{|c|c|c|}
\hline 
$A\backslash B$ & $p_{M}$ & $p_{C}$\tabularnewline
\hline 
$p_{M}$ & $1,\,1$ & $-y,\,1+x$\tabularnewline
\hline 
$p_{C}$ & $1+x,\,-y$ & $0,\,0$\tabularnewline
\hline 
\end{tabular}\caption{\label{tab:PD_equivalent} Normalizing and re-writing the payoff matrix in Table \ref{tab:PD_payoffs} in terms of $x$ and $y$.}
\end{table}
\medskip{}

\subsection{Algorithms and equilibria}

With two prices we can easily describe all possible algorithms that
the sellers could use, see Table \ref{tab:PD_algorithms}. Here, $s_C$ represents the algorithm that chooses price $p_{C}$ no matter what the competitor does, and similarly $s_{M}$ is the algorithm that always chooses $p_M$. Further, $s_{T}$ is the tit-for-tat algorithm that replicates the competitor's action, and $s_{R}$ is the reverse-tit-for-tat that chooses the exact opposite of the competitor's action.

\begin{table}
\centering{}%
\begin{tabular}{|c|l|l|}
\hline 
$s_{C}$ & $p_{C},p_{M}\rightarrow p_{C}$ & always ``competitive'' price\tabularnewline
\hline 
$s_{M}$ & $p_{C},p_{M}\rightarrow p_{M}$ & always ``monopoly'' price\tabularnewline
\hline 
$s_{T}$ & $p_{C}\rightarrow p_{C},p_{M}\rightarrow p_{M}$ & tit-for-tat\tabularnewline
\hline 
$s_{R}$ & $p_{C}\rightarrow p_{M},p_{M}\rightarrow p_{C}$ & ``reverse'' tit-for-tat\tabularnewline
\hline 
\end{tabular}\caption{\label{tab:PD_algorithms}Possible algorithms for two prices}
\end{table}

Given four possible choices of algorithms, we can completely describe convergent outcomes, this is reported in Table \ref{tab:PD_outcomes}. The entries in the table report the possible final prices for any pairs of algorithms, $\xi(s^{-i},s^{i})$. For example, when one of the algorithms is $s_{T}$ and the other is $s_{R}$ we
end up with a cycle (which is ruled out as a final outcome by assumption). When both
algorithms are $s_{R}$, the seller who was the last to adjust their
algorithm can the choose initial price sequence and, thus, out of the two possible
final prices pairs $(p_{C},p_{M})$ and $(p_{M},p_{C})$ we choose
the one that is preferred by this seller.

\begin{table}
\centering{}{\scriptsize{}}%
\begin{tabular}{|c|c|c|c|c|}
\hline 
 & {\scriptsize{}$s_{M}$} & {\scriptsize{}$s_{C}$} & {\scriptsize{}$s_{T}$} & {\scriptsize{}$s_{R}$}\tabularnewline
\hline 
{\scriptsize{}$s_{M}$} & {\scriptsize{}$p_{M},p_{M}$} & {\scriptsize{}$p_{M},p_{C}$} & {\scriptsize{}$p_{M},p_{M}$} & {\scriptsize{}$p_{M},p_{C}$}\tabularnewline
\hline 
{\scriptsize{}$s_{C}$} & {\scriptsize{}$p_{C},p_{M}$} & {\scriptsize{}$p_{C},p_{C}$} & {\scriptsize{}$p_{C},p_{C}$} & {\scriptsize{}$p_{C},p_{M}$}\tabularnewline
\hline 
{\scriptsize{}$s_{T}$} & {\scriptsize{}$p_{M},p_{M}$} & {\scriptsize{}$p_{C},p_{C}$} & {\scriptsize{}$p_{M},p_{M}$} & {\tiny{}$\begin{array}{l}
p_{C},p_{C}\rightarrow p_{C},p_{M}\rightarrow\\
p_{M},p_{M}\rightarrow p_{M},p_{C}
\end{array}$}\tabularnewline
\hline 
{\scriptsize{}$s_{R}$} & {\scriptsize{}$p_{C},p_{M}$} & {\scriptsize{}$p_{M},p_{C}$} & {\tiny{}$\begin{array}{l}
p_{C},p_{C}\rightarrow p_{M},p_{C}\rightarrow\\
p_{M},p_{M}\rightarrow p_{C},p_{M}
\end{array}$} & {\scriptsize{}$p_{C},p_{M}$}\tabularnewline
\hline 
\end{tabular}\caption{\label{tab:PD_outcomes}Final outcomes for different pairs of algorithms assuming row player was the last to choose algorithm}
\end{table}

Now, we can state the main result that completely characterizes equilibria for this two-price game. For starters, an equilibrium always exists. In addition, the universe of equilibria takes a simple form in that it all boils down to three parameters: $x$, $y$ and $\beta$, and we can characterize this set precisely. To simplify notation, we write $s^{-i}\rightarrow s_{i}$ to mean $f(s^{-i}) = s^{i}$, that is the strategy $s_{i}$ optimally chosen by seller $i$ when the competitor is playing $s^{-i}$. 

\begin{theorem}
\label{th:PD_equilibria}$\ $
\begin{enumerate}
\item A Markov equilibrium always exists.  
\item If $x\le\beta$ only the following Markov equilibrium is possible:
\[
\begin{array}{l}
s_{R}\rightarrow s_{C}\rightarrow s_{T}\\
s_{T}\rightarrow\{s_{T},s_{M}\}\\
s_{M}\rightarrow\{s_{T},s_{M}\}
\end{array}\qquad(type\ I)
\]
\item If $x>\beta$ then 
\begin{enumerate}
\item the following equilibrium is always possible: 
\[
s_{R},s_{M}\rightarrow s_{C}\rightarrow s_{T}\rightarrow s_{T}\qquad(type\ II)
\]
\item in addition if $y<\beta\cdot(x-\beta)$ the following equilibrium
is also possible: 
\[
\begin{array}{l}
s_{M},s_{R}\rightarrow s_{R}\\
s_{C},s_{T}\rightarrow s_{T}
\end{array}\qquad(type\ III)
\]
\end{enumerate}
\end{enumerate}
\end{theorem}
\begin{figure}
\begin{centering}
\includegraphics[height=10cm]{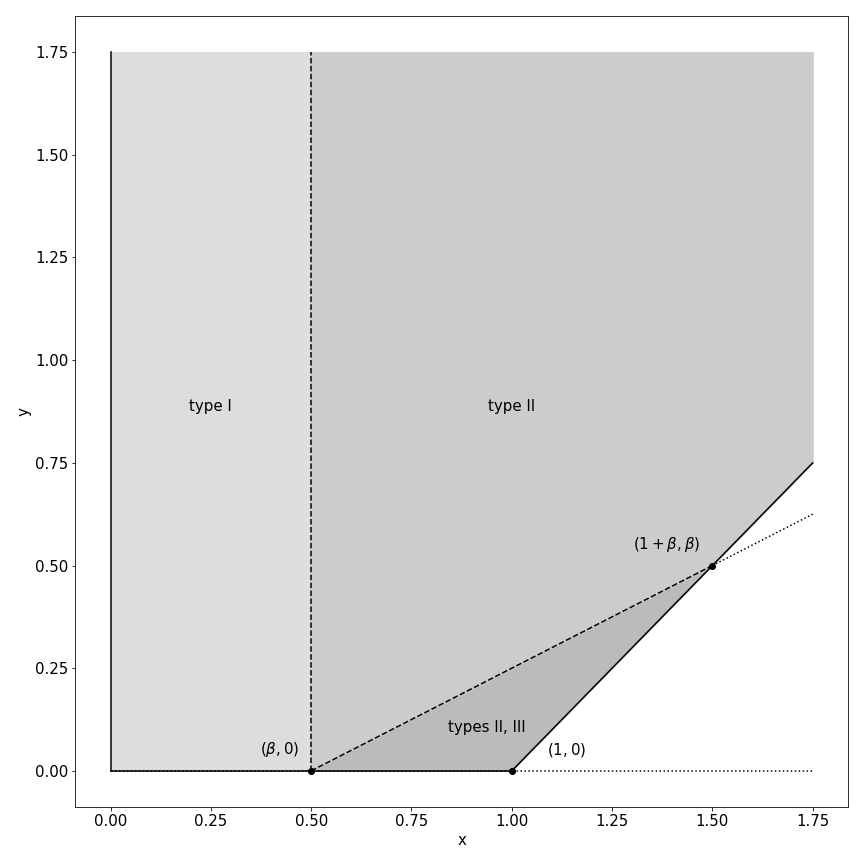}
\par\end{centering}
\caption{\label{fig:PD_equilibria}Possible equilibria for varying values of $x$ and $y$, $\beta$ fixed at 0.5.}
\end{figure}

Theorem \ref{th:PD_equilibria} describes which equilibria are
possible depending on structure of the payoff matrix (i.e., $x$ and $y$) and effective discount factor (i.e. $\beta$). For example, if $\beta$ is lower than $x$, then in the type II equilibrium, for algorithm $s_{R}$ or $s_{M}$ chosen by the opponent, it is optimal for the seller to set the algorithm $s_{C}$. This is because for a low enough discounting, the player choosing the algorithm cares about the immediate payoff and thus wants prices to converge to $(p_C,p_M)$ in the short run. However, going further, when it is the opponent's turn to react to $s_C$, they choose $s_T$ for reasons explained in the introduction; and so on. Figure \ref{fig:PD_equilibria} illustrates the result graphically---fixing $\beta$ it partitions the $(x,y)$ space into the three types of equilibria reported in Theorem \ref{th:PD_equilibria}.

The proof of Theorem \ref{th:PD_equilibria} takes a general approach as it treads through the set of best responses for each chosen algorithm of the opponent. We first show that it is always optimal to respond to $s_C$ with $s_T$. To do that we consider alternative responses and equilibrium sequences of algorithms that such responses can generate and show that in each of them it would be optimal for at least one of the sellers to deviate from the current strategy. Then, we show after $s_T$, the best response is either $s_T$ or $s_M$ and  we can only observe monopoly prices $(p_M,p_M)$ even after subsequent opportunities to revise algorithms. Finally, we find the best responses to $s_R$ and $s_M$ by comparing the expected payoffs for each of the alternatives.

What are the  pricing outcomes corresponding to these equilibria? The next result documents that the answer to this follows immediately from matching the three types of equilibria reported in Proposition \ref{th:PD_equilibria} with possible final outcomes reported in Table \ref{tab:PD_outcomes}. It states that for equilibria for types I and II, the unique outcome is the monopoly price $(p_{M},p_{M})$, and for type III equilibrium prices oscillate between $(p_{C},p_{M})$ and $(p_{M},p_{C})$.

\begin{corollary}
\label{cor:PD_prices}$\ $
\begin{enumerate}
\item In any equilibrium we either converge to a monopoly price pair $(p_{M},p_{M})$
(for equilibria of types I or II) or end up with an alternating sequence
of prices
\begin{equation}
(p_{C},p_{M})\rightarrow(p_{M},p_{C})\rightarrow(p_{C},p_{M})\rightarrow\dots
\end{equation}
in which prices change every time one of the sellers gets the opportunity
to adjust their algorithm (for equilibrium of type III).
\item Type III equilibrium (and the corresponding alternating price sequence) is only possible if
\[
x<2,\qquad y<\frac{x^{2}}{4}, \text{ and} \qquad \frac{x}{2}-\sqrt{\frac{x^{2}}{4}-y}<\beta<\frac{x}{2}+\sqrt{\frac{x^{2}}{4}-y}. 
\]
\end{enumerate}
\end{corollary}
\begin{figure}
\begin{centering}
\includegraphics[height=10cm]{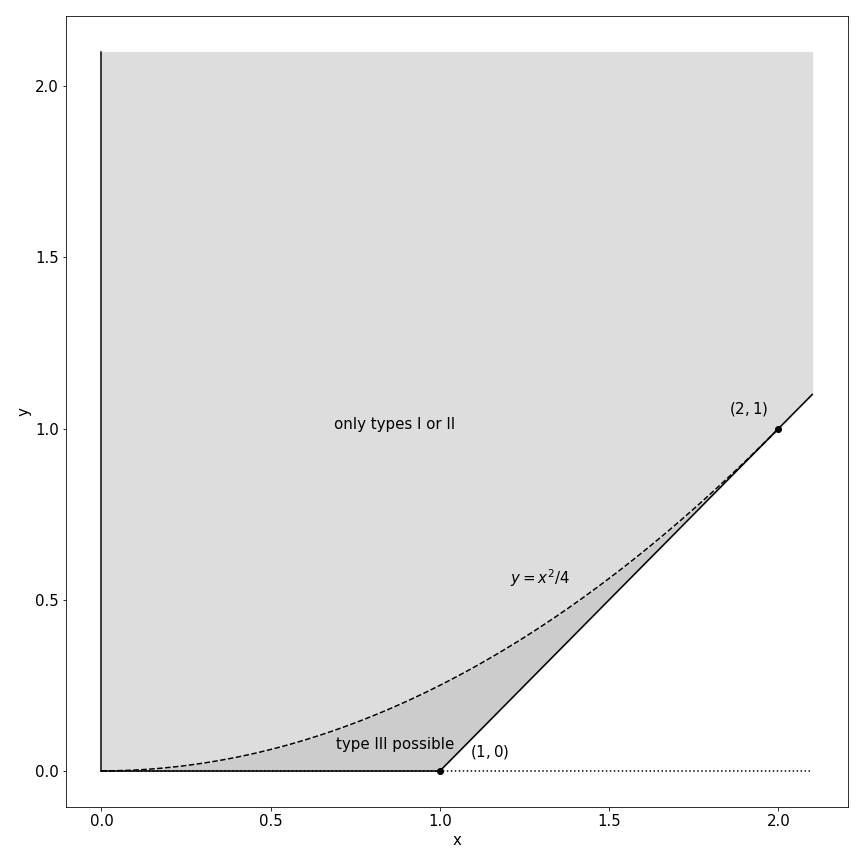}
\par\end{centering}
\caption{\label{fig:PD_alternating prices}When type III equilibrium with alternating prices can arise for some value of $\beta$.}
\end{figure}

The main message in Corollary \ref{cor:PD_prices} is that some level of market power is inevitable, for $(p_C,p_C)$ is never a convergent equilibrium outcome. Moreover, in {\it most} cases, the repeatedly play of the pure monopoly price is the unique equilibrium outcome. This claim can be visualized in Figure \ref{fig:PD_alternating prices}, which shades the values of $x$ and $y$ for which type III equilibrium is possible at least for some values of the discount factor $\beta$. Note that ``for some $\beta$" is a fairly weak criterion, and still the region is quite small. This leads us to conclude that equilibria of types I and II dominate the parameter space and hence the unique monopoly outcome is to expected in such two-price games. 

As a final step, we provide some simple sufficient conditions on the payoff matrix and the profit function under which repeated play of the monopoly price is the unique equilibrium outcome. 
\begin{corollary}
\label{cor:PD_monopoly_condition}$\ $
\begin{enumerate}
\item If
\[
y>x\quad\text{ that is } \quad \pi(p_{M},p_{C})+\pi(p_{C},p_{M})<\pi(p_{M},p_{M})+\pi(p_{C},p_{C}) \]
then only equilibria of type I or II are possible, which means that
prices always converge to a monopoly outcome $(p_{M},p_{M})$.
\item If payoff function $\pi(p,q)$ is quadratic (demand is linear):
\[
\pi(p,q)=p\cdot\frac{D-p+\alpha(q-p)}{2D}
\]
or if $\pi$ is a ``discrete choice'' profit function:
\[
\pi(p,q)=p\cdot\frac{e^{-bp}}{a+e^{-bp}+e^{-bq}}
\]
then again in any equilibrium, prices always converge to a monopoly outcome
($p_{M},p_{M})$.
\end{enumerate}
\end{corollary}

That $y>x$ leads uniquely to the monopoly outcome is already visible in Figure \ref{fig:PD_alternating prices}. Conceptually, a quick look at Table \ref{tab:PD_equivalent} tells us that if the gain from ``deviating" from $(p_M,p_M)$ is smaller than the loss the opponent would incur from this deviation, then the opponent cannot be incentivized with continuation value for any value of $\beta$. 

Part 2 of the result shows that for two widely used profit functions---the linear demand case invoked in various textbooks and the discrete choice model used often in applied work, repeatedly play of the monopoly outcome is again the unique equilibrium outcome. The proof essentially shows that for these profit functions the condition $y>x$ is always satisfied, hence only type I or II equilibria arise. 


\section{Multiple prices} \label{section multiple}

\subsection{Main result}

In this section we consider the case in which for each seller there
are more than two possible prices: $p_{C},p_{C+\Delta},\dots,p_{M-\Delta,}p_{M}$.
With multiple prices the number of possible algorithms explodes (there
are $3^{3}=27$ algorithms for 3 prices, and $4^{4}=256$ algorithms
for 4 prices and so on). Thus, it is now difficult to ``brute force" our way through all best response functions to describe Markov equilibria. 

With multiple prices there also exist {\it unequal} equilibria, in which one seller makes a profit even higher than the monopoly outcome.
Suppose $(p,q)$ is a pair of prices such that payoffs of the sellers at these prices $(\pi(p,q),\pi(q,p)$)
are on the convex hull of the Pareto Frontier and in addition that
\begin{equation}
\pi(p_{C},p_{C})<\pi(q,p)<\pi(p_{M},p_{M})<\pi(p,q).
\end{equation}
So, the payoff of one seller is above the monopoly outcome and the other is between the monopoly and competitive outcome. Now consider the following algorithms
\begin{equation}
\begin{array}{l}
s^{i}=\{q\rightarrow p,x\rightarrow p_{C}\,for\,x\ne q\}\\
s^{-i}=\{p\rightarrow q,x\rightarrow p_{C}\,for\,x\ne p\}
\end{array}
\end{equation}
This is a grim trigger sort of strategy that aims to sustain the choice of $(p,q)$ in equilibrium. If the discount factor $\beta$ is sufficiently low these two algorithms
form an equilibrium pair (for seller $i$ it is optimal
to respond with $s^{i}$ to $s^{-i}$ and the vice-versa). For
sufficiently low $\beta$: $\xi(s^{i},f^{-i}(s^{i}))=(p,q)$ and $\xi(s^{-i},f^{i}(s^{-i}))=(q,p)$ in any subgame perfect (and hence Markov perfect) equilibrium. Thus, each seller can guarantee a payoff of at least $\pi(p,q)$ or $\pi(q,p)$ in response to 
$s^{-i}$ and $s^i$ correspondingly. Finally, since the continuation payoffs $(\pi(p,q),\pi(q,p)$) are
on the convex hull of the Pareto frontier, none of the sellers would
be able to get a higher payoff by deviating from the current strategy.

Interestingly, such unequal equilibria can also be constructed for arbitrary high discount factors $\beta$. If seller $-i$ believes that for any of their response seller $i$ will keep their current strategy $s^{i}$, then they also do not have any incentives to deviate from $s^{-i}$. And for seller $i$ it would indeed be optimal to do that if $\pi(p,q)$ is sufficiently high. To illustrate these ideas, we explicitly construct such an equilibrium with 5 different prices (see Part D of the appendix). As in repeated games, it is however easier to operate in the payoff space and characterize key properties of the set of equilibrium payoffs, which is what we do in this section. 

The obvious question to then ask is this: Can we bound the equilibrium payoffs to illustrate the ideas of the inevitability of some collusion and commonality of high collusion? The main result shows that for any Markov equilibrium the payoff of both sellers is bounded from below by the competitive outcome $\pi(p_{C},p_{C})$, and in any equilibrium at least one seller gets a payoff that is eventually bounded from below by a number that is approximately equal to the monopoly profit $\pi(p_{M},p_{M})$.

\begin{theorem}
\label{prop:multiple_prices}Suppose $s_{1},s_{2},s_{3},\dots$ is
a path of algorithms in a Markov equilibrium. Then:
\begin{enumerate}
    \item payoff of both sellers is higher than the competitive outcome: 
    \[\forall\text{ } k\geqslant 1, \{u^{i}_{k}, v^{i}_{k+1}\}\geqslant \pi(p_{C},p_{C}) \text{ } \forall \text{ } i \in \{A,B\}, \text{ and}\]
    \item payoff of at least one seller is eventually bounded below approximately by the monopoly outcome: 
    \[\exists \text{ } k \text{ s.t. } \max_{m\geqslant k}\{u^{i}_{m}, v^{i}_{m+1}\}\geqslant (1-\beta)\cdot\pi(p_{M-\Delta},p_{M-\Delta})+\beta\cdot \pi(p_{M},p_{M}) \text{ for $i=A$ or $i=B$}.\]
\end{enumerate}
\end{theorem}
While the first part is straightforward, the second one can be interpreted as follows: If the price grid is dense enough or the effective discount factor is high enough, then in any Markov equilibrium, the payoff of at least one seller is bounded from below by the monopoly profit. Both point to the exercise of some market power by the sellers and taken together they establish this market power is significant.  Notice, also, that even for the weakest case of $\beta\approx 0$ and only 3 possible prices in the grid the lower bound it at least $\pi\left(\frac{p_C+p_M}2,\frac{p_C+p_M}2\right)$ which is non-trivially higher than the competitive payoff.\footnote{The reason why the lower bound is not exactly the monopoly payoff, is that for general payoff matrices with multiple prices there could also exist alternating price equilibria (similar to type III equilibria from Theorem \ref{th:PD_equilibria}). We believe that for specific payoffs functions (e.g. "linear" or "discrete choice") such equilibria might be also ruled out and the lower bound could be improved.}

\subsection{Proof of Theorem \ref{prop:multiple_prices}}

The proof consists of several parts. In sections \ref{subsec:state_space}-\ref{subsec:best_payoffs} we develop the necessary techniques. In \ref{subsec:state_space} we show that the state space can be contracted from the set of algorithms to the matrix of price pairs,  which significantly reduces the complexity of the analysis. Then in \ref{subsec:payoff_matrices} we derive the conditions that the continuation payoff matrices should satisfy in order for the them to correspond to an equilibrium. Section \ref{subsec:best_payoffs} then defines a notion of `best'' and ``second best'' algorithms that each seller could have whenever their competitor is adjusting their algorithm. 

Finally, we prove the theorem in section \ref{subsec:continuation}. We consider an arbitrary equilibrium path, and show that the expected continuation payoffs for both sellers tend to increase on this path and if not the corresponding seller can guarantee at worst the ``second best'' continuation payoff. Lemma \ref{lem:v_SB} then shows that such payoff is above the limit stated in Theorem \ref{prop:multiple_prices} for at least one of the sellers. 

\subsubsection{State Space Reduction \label{subsec:state_space}} 

Whenever a seller responds with an algorithms that leads to a price
pair $(p,q)$ they could choose always the best algorithm consistent
with this price pair (that maps $q$  into $p$). This means that whenever we observe a convergent
prices pair $(p,q)$ after seller $i$ adjusted their algorithm, we
should observe the same algorithm for seller $i$. Thus, the price
pair $(p,q)$ becomes essentially the state variable, and since there
are much fewer prices pairs than possible algorithms this allows us
to significantly reduce the state space. 

We can define
\begin{equation}
v^{i}(p,q)=\max_{s:q\rightarrow p}v^{i}(s)
\end{equation}
\begin{equation}
s^{i}(p,q)=\arg\max_{s:q\rightarrow p}v^{i}(s)
\end{equation}
\begin{equation}
u^{-i}(q,p)=u^{-i}(s^{i}(p,q))
\end{equation}
If we know the payoff matrices $v^{i}(p,q)$ and $v^{-i}(q,p)$ then
for any algorithm $s$ of seller $-i$ we can calculate the
highest possible continuation payoff that seller $i$ could get:
\begin{equation}
u^{i}(s)=\max_{p}(1-\beta)\cdot\pi(p,s(p))+\beta\cdot v^{i}(p,s(p))\label{eq:u}
\end{equation}
Even though this equation does not tell us with which
exact algorithms seller $i$ should respond, it pins down the highest payoff they can ensure while responding to an algorithm $s$ by $-i$. 

\subsubsection{Payoff Matrices \label{subsec:payoff_matrices} in Equilibrium}

In this section we will derive the conditions that matrices
$v^{i}(p,q)$ and $u^{-i}(q,p)$ for $i=A,B$ should satisfy if they
correspond to payoffs in a certain Markov equilibrium. Also, we will
show how to construct the best response to any algorithm if the values
of the payoff matrices are known. 

First, let us define the worst possible algorithm for seller $i$
to respond to, and the lowest continuation payoff that they can get when
responding to some algorithm:
\begin{definition}
The worst algorithm of seller $-i$ for seller $i$ is defined as:
\begin{equation}
\underline{s}^{-i}(p)=\arg\min_{q}(1-\beta)\cdot\pi(p,q)+\beta\cdot v^{i}(p,q)
\end{equation}
\end{definition}
\begin{definition}
Then the lowest possible continuation payoff of $i$ when $i$ is
adjusting their algorithm is:
\begin{equation}
\underline{u}^{i}=\max_{p}\left[\min_{q}(1-\beta)\cdot\pi(p,q)+\beta\cdot v^{i}(p,q)\right]
\end{equation}
\end{definition}
From equation (\ref{eq:u}) it follows that for any algorithm $s$:
\begin{equation}
u^{i}(s)\ge\underline{u}^{i}=u^{i}(\underline{s}^{-i})
\end{equation}
Now let's calculate $v^{i}(p,q)$. Suppose $i$ chooses
an algorithm such the next pair of prices (after $-i$ adjusts their
algorithm) is $(p',q')$. Then for this price pair we can define
the ``before'' continuation payoffs $V^{i}(p',q')$ and $U^{-i}(q',p')$,
which include this price pair and assume optimal behavior afterwards:
\begin{definition}
``Before'' continuation payoffs for price pair ($p',q')$ are defined
as
\begin{equation}
\begin{array}{l}
V^{i}(p',q')=(1-\beta)\cdot\pi(p',q')+\beta\cdot u^{i}(p',q')\\
U^{-i}(q',p')=(1-\beta)\cdot\pi(q',p')+\beta\cdot v^{-i}(q',p')
\end{array}
\end{equation}
\end{definition}

\noindent Then, if price pair $(p',q')$ follows $(p,q)$ the continuation payoffs
are linked as:
\begin{equation}
v^{i}(p,q)=V^{i}(p',q')\qquad u^{-i}(q,p)=U^{-i}(q',p')
\end{equation}
Note that the values of the ``before'' continuation payoffs $V^{i}(p',q')$
and $U^{-i}(q',p')$ are directly determined by the ``after'' continuation
payoff matrices $u^{i}(p',q')$ and $v^{-i}(q',p')$. Finally, the
following proposition links the ``before'' and ``after'' continuation
payoffs, thus establishing the recursive relationship that the continuation
payoff matrices $u^{i}(p',q')$ should satisfy.

\begin{proposition}
\label{prop:before_after}For any price pair $(p,q)$:
\begin{equation}
\label{eq:before_after}
v^{i}(p,q)=\max V^{i}(p',q')\qquad s.t.\ \begin{array}{l}
U^{-i}(q',p')\ge\underline{u}^{-i}\\
U^{-i}(q',p')\ge U^{-i}(q,p)\\
either\,(p',q')=(p,q)\,or\,q'\ne q
\end{array}
\end{equation}
Moreover, if the price pair $(p',q')$ maximizes the above expression
then $u^{-i}(q,p)=U^{-i}(q',p')$.
\end{proposition}

\medskip

\noindent The idea of the proof is as follows. If price pair $(p',q')$ goes after $(p,q)$ then all conditions stated in equation (\ref{eq:before_after}) should be satisfied (the payoff of seller $-i$ can never be below $\underline{u}^{-i}$ and it cannot be below $U^{-i}(q,p)$ since $-i$ can keep the strategy from the previous period). In addition (see appendix for the details), by using the following algorithm:
\begin{equation}
\label{eq:transition_algorithm}
s_{*}^{i}(p,q)=\{q\rightarrow p,q'\rightarrow p',x\rightarrow\underline{s}^{i}(x)\,for\,x\ne q\,and\,x\ne q'\}
\end{equation}
seller $i$ can make sure that the subsequent prices will actually be $(p',q')$.

Proposition \ref{prop:before_after} gives us a powerful tool showing which price transitions are possible in an equilibrium. In addition, it gives us precise conditions that payoff matrices $v^i(p,q)$ and $u^{-i}(q,p)$ should satisfy in order for them to correspond to an equilibrium.

\subsubsection{``First best'' and ``Second best'' Continuation Payoffs \label{subsec:best_payoffs}}

The final piece of notation that we would need in order to prove Theorem \ref{prop:multiple_prices} is the first best and the second best continuation payoffs and the corresponding algorithms that would allow sellers to achieve such payoffs. 

\begin{definition}
\label{def:first_best}
The ``best" continuation payoff of $i$ when $-i$ is adjusting their
algorithm is defined as:
\begin{equation}
v_{*}^{i}=\max V^{i}(p,q)\qquad s.t.\ U^{-i}(q,p)\ge\underline{u}^{-i}\label{eq:first_best}
\end{equation}
Also if $(p_{*}^{i},q_{*}^{i})$ maximize the above expression then we define the best continuation algorithm of $i$ as
\footnote{In Definitions \ref{def:first_best} and \ref{def:second_best} if several feasible price pairs maximize $V^i(p,q)$ then in accordance with assumption 3 from Section \ref{subsec:assumptions} we select the one maximizing $U^{-i}(q,p)$.}:
\begin{equation}
s_{*}^{i}=\{q_{*}^{i}\rightarrow p_{*}^{i},x\rightarrow\underline{s}^{i}(x)\,for\,x\ne q_{*}^{i}\}
\end{equation}
\end{definition}
\begin{definition}
\label{def:second_best}
If $(p_{*}^{i},q_{*}^{i})$ are the ``first best" prices from definition \ref{def:first_best}
then we define the ``second best'' continuation payoff of seller
$i$ when $-i$ is adjusting their algorithm as:
\begin{equation}
v_{**}^{i}=\max V^{i}(p,q)\qquad s.t.\ \begin{array}{l}
U^{-i}(q,p)\ge\underline{u}^{-i}\\
q\ne q_{*}^{i}
\end{array}\label{eq:second_best}
\end{equation}
Also if $(p_{**}^{i},q_{**}^{i})$ maximizes the above expression
(\ref{eq:second_best}) then we can define the ``second best'' continuation
algorithm of $i$ as:
\begin{equation}
s_{**}^{i}=\{q_{**}^{i}\rightarrow p_{**}^{i},x\rightarrow\underline{s}^{i}(x)\,for\,x\ne q_{**}^{i}\}
\end{equation}
\end{definition}

Note that seller $i$ cannot get a continuation payoff $v^i(s)$ higher than $v_*^i$ for any $s$. Moreover, the payoff $v_*^i$ is exactly achievable for $s = s_{*}^{i}$ (if  seller $-i$ does not continue with prices $(p_*^i,q_*^i)$ their payoff is at most $\underline{u}^{-i}$). The second best payoff
$v_{**}^{i}$ corresponds to the case when seller $i$ cannot use the opponent's price $q_{*}^{i}$ from the first best pair $(p_*^i,q_*^i)$.\footnote{If none of the price pairs satisfy the stated conditions for the ``second-best" (this means seller $-i$ always ends up with price $q_*$ after adjusting their algorithm), it is fairly easy to show that the continuation payoffs for at least one of the sellers in any equilibrium sequence will exceed $\pi(p_M,p_M)$. See part B of the appendix for details. So, in what follows we will assume that the second-best always exists.}

\subsubsection{Evolution of Equilibrium Continuation Payoffs \label{subsec:continuation}}

Consider some arbitrary equilibrium sequence of algorithms $s_{0},s_{1},s_{2},\dots$
Note that ``even'' algorithms in the sequence $s_{0},s_{2},\dots$ correspond
to one seller, and ``odd'' algorithms to the other. Below
when we write expressions like $u^{i}(s_{k})$ we assume that $i$ and $k$
agree with each other (e.g. in this case it means that $s_{k}$ is
the algorithm of seller $i$).

Lemma \ref{lem:uv_pC} shows that each seller can guarantee that their the continuation payoffs stay above the competitive payoff $\pi(p_C,p_C)$ which proves the first part of Theorem \ref{prop:multiple_prices}.

\begin{lemma}
\label{lem:uv_pC} $u^{i}(s_{k})\ge\pi(p_{C},p_{C})$ for $k\ge0$ and $v^{-i}(s_{k})\ge\pi(p_{C},p_{C})$ for $k\ge1$.
\end{lemma}

We can also show that the continuation payoffs tend to increase on the equilibrium path (and if not the corresponding seller can guarantee at least the second best continuation payoff).

\begin{lemma}
\label{lem:uv}For any $k$ and corresponding i: 
\begin{enumerate}
\item $u^{-i}(s_{k+1})\ge v^{-i}(s_{k})$
\item Either $v^{i}(s_{k+1})\ge u^{i}(s_{k})$ or $v^{i}(s_{k+1})=v_{*}^{i}$ or $v^{i}(s_{k+1})=v_{**}^{i}$. Moreover, the last case is only possible if  $\xi(s_{k},s_{k+1})=(q_{*}^{i},x)$
for some $x$ where $q_{*}^{i}$ is from definition \ref{def:first_best}. 
\end{enumerate}
\end{lemma}

Suppose the continuation payoffs of one seller becomes constant, then we can show that the continuation payoffs of the other seller will also become constant. Moreover, at least on of these constant payoffs should be above the monopoly profit.

\begin{lemma}
\label{lem:constant_payoffs}
If for some $i$ and $k$:
\begin{equation}
u^{i}(s_{k})=v^{i}(s_{k+1})=u^{i}(s_{k+2})=\dots
\end{equation}
 then for some $m$
\begin{equation}
u^{-i}(s_{m})=v^{-i}(s_{m+1})=u^{-i}(s_{m+2})=\dots
\end{equation}
Moreover, either $u^{i}(s_{k})\ge \pi(p_M,p_M)$ or $u^{-i}(s_m)\ge\pi(p_M,p_M).$
\end{lemma}

Finally, if continuation payoffs for both sellers never become constant then eventually both of them will be above the corresponding second best payoffs $v_{**}^i$. Thus, to prove theorem \ref{prop:multiple_prices} we only need to show that at least one of these second best payoffs is above the bound stated in the theorem.

\begin{lemma}
\label{lem:v_SB}$\ $
\begin{equation}
\max\left(v_{**}^{i},v_{**}^{-i}\right)\ge(1-\beta)\cdot\pi(p_{M-\Delta},p_{M-\Delta})+\beta\cdot\pi(p_{M},p_{M})
\end{equation}
\end{lemma}

\section{Extensions and robustness considerations \label{sec:extensions}}

The main message of this section is that if the reader is worried how the main results would extend to scenarios where Assumptions 1 and 2 are relaxed; well, they shouldn't be too worried. If the reader ponders what would happen if we extend the equilibrium notion from Markov perfect to subgame perfect; the sharpness of results is tempered by the multitude of equilibria and quite complex strategies are required to sustain these. 

\subsection{Alternative Ways to Resolve Cycles } \label{resolving cycles}

For our main results we assumed that the sellers cannot respond with algorithms that generate cycles (Assumption 2 in Section \ref{subsec:assumptions}). This is essentially a restriction on the strategy space. In this section we relax this assumption for the two prices model, presented in Section \ref{sec:PD}. Recollect from Table \ref{tab:PD_outcomes} that whenever $s_{T}$ is chosen by one seller and $s_{R}$ by the other, prices go into a cycle. We consider two ways of resolving cycles:
\begin{enumerate}
\item A cycle is equivalent to having the lowest prices from each seller in the cycle. This is motivated by the fact that if prices change frequently then customers can learn about this and time their purchase so that they can get the lowest possible price:
\begin{equation}
\pi(s_{T},s_{R})=\pi(s_{R},s_{T})=\pi(p_{C},p_{C})=0
\end{equation}
\item A seller's payoff from a cycle is equivalent to the average payoff from all price pairs that are observed in the cycle. This means customers arrive continuously and do not optimize on timing:
\begin{equation}
\begin{array}{ll}
\pi(s_{T},s_{R})=\pi(s_{R},s_{T}) &=\frac{\pi(p_{C},p_{C})+\pi(p_{M},p_{C})+\pi(p_{M},p_{M})+\pi(p_{C},p_{M})}{4} \\
&=\frac{2+x-y}{4}<\frac{3}{4}
\end{array}
\end{equation}
\end{enumerate}
We can first show that for both options it is never optimal to respond to $s_T$ with $s_R$. 
\begin{lemma}
\label{lem:sTsR}$f^{i}(s_{T})\ne s_{R}$ $\forall$ $i\in\{A,B\}$.
\end{lemma}
\begin{proof}
See appendix.$\ $
\end{proof}

Lemma \ref{lem:sTsR} implies that the first half of the proof of Theorem \ref{th:PD_equilibria} is unchanged, which means the best response to $s_{C}$ can still only be $s_{T}$ and after $s_{T}$ we still observe only
monopoly prices. Now it is not difficult to describe all possible equilibria with cycles.

\begin{proposition}
\label{prop:PD_equilibria_cycles}$\ $
\begin{enumerate}
\item If cycles are equivalent to the lowest price (case 1) then all equilibria are exactly the same same as described in Theorem \ref{th:PD_equilibria}
\item If cycles are equivalent to the average payoff (case 2) then
\begin{enumerate}
\item for $x\le\beta$ and $y\ge4\beta-2-3x$ the following Markov equilibria are possible
\begin{equation}
\begin{array}{l}
s_{R}\rightarrow s_{C}\rightarrow s_{T}\\
s_{T}\rightarrow\{s_{T},s_{M}\}\\
s_{M}\rightarrow\{s_{T},s_{M}\}
\end{array}\qquad(type\,I)
\end{equation}
\item for $x\le\beta$ and $y\le4\beta-2-3x$ the following Markov equilibria are possible
\begin{equation}
\begin{array}{l}
s_{R}\rightarrow s_{T}\\
s_{T}\rightarrow\{s_{T},s_{M}\}\\
s_{M}\rightarrow\{s_{T},s_{M}\}
\end{array}\qquad(type\,I')
\end{equation}
\item for $x>\beta$ the equilibria are the same as described in Theorem \ref{th:PD_equilibria}.
\end{enumerate}
\end{enumerate}
\end{proposition}

Thus, the only difference from Theorem \ref{th:PD_equilibria}  is that in case 2 for some values of parameters it might be optimal to switch directly to $s_{T}$ from $s_{R}$ without the intermediate step of $s_{C}$ in an equilibrium
of type I. Type II and III equilibria, where the latter can lead to alternating price sequences, exist for exactly the same values of parameters $x$ and $y$. And, as before, for majority of the parameters monopoly price is the unique convergent outcome in equilibrium.

\subsection{Precise Calculations of Expected Payoffs} \label{sec:precise}

In this extension we will relax Assumption 1 from Section \ref{subsec:assumptions} and
show that our results generally hold (for sufficiently small interval
between price adjustments $\Delta t$) even if we calculate payoffs
precisely without ignoring initial convergence and discreteness of
price adjustments. First we rewrite the precise model incorporating pre-convergent prices and calculate expected payoffs within this model. Then we study the Markov equilibria of this updated framework.

Since the updated model and corresponding results require the introduce of new notation, we relegate the formalization to Part B of the appendix. The main insight can be (informally) stated thus:

\begin{claim*}
For sufficiently small $\Delta t$, every Markov equilibrium in the ``precise model" generically corresponds a Markov equilibrium of the ``main model" we study by resolving the dependence on the initial prices arbitrarily.
\end{claim*}

\subsection{Subgame Perfect Equilibria \label{subsec:SPE}}

In this extension we consider the more general subgame perfect equilibria (SPE) and try to give a sense of the possible outcomes. We restrict our attention to the simplest two price case. The sellers still adjust their algorithms asynchronously, however, now the response of each might depend not only on the current algorithm of the opponent, but on the overall history of algorithms up to this
point.

Suppose one of the sellers is adjusting their
algorithm and suppose the current algorithm of their opponent is $s$.
Define the set of all possible continuation payoffs ($u,v$)
in a subgame perfect equilibrium depending on $s$ ($u$
corresponds to the seller adjusting their algorithm and $v$ to the
other one):
\begin{equation}
H(s)=\{(u,v):\ (u,v)\ is\ possible\ in\ SPE\ starting\ with\ s\}
\end{equation}

There is a recursive way to construct these sets in the spirit of \citet*{aps}. Suppose the seller adjusting their algorithm decides to respond to $s$ with $s'$, then
the set of possible continuation payoffs is as follows:
\begin{equation}
H(s,s')=\left\{ (u,v):\begin{array}{l}
u=(1-\beta)\cdot\pi(s',s)+\beta\cdot v'\\
v=(1-\beta)\cdot\bar{\pi}(s,s')+\beta\cdot u'\\
(u',v')\in H(s')
\end{array}\right\} 
\end{equation}
The lowest possible continuation payoff $u$ of the seller adjusting
their algorithms is then:
\begin{equation}
u_{min}(s)=\max_{s'}\min_{(u,v)\in H(s,s')}u
\end{equation}
Finally the set of possible SPE payoffs starting with
$s$ is:
\begin{equation}
H(s)=F(H)(s)=\bigcup_{s'}\left\{ (u,v):\begin{array}{l}
(u,v)\in H(s,s')\\
u\ge u_{min}(s)
\end{array}\right\} \label{eq:mapping}
\end{equation}
Any $(u,w)$ with $u\ge u_{min}(s)$ could be sustained in an SPE since for any deviation their exists a subgame which generates a payoff not higher than $u$.
\begin{figure}
    \centering
   \includegraphics[scale=0.35]{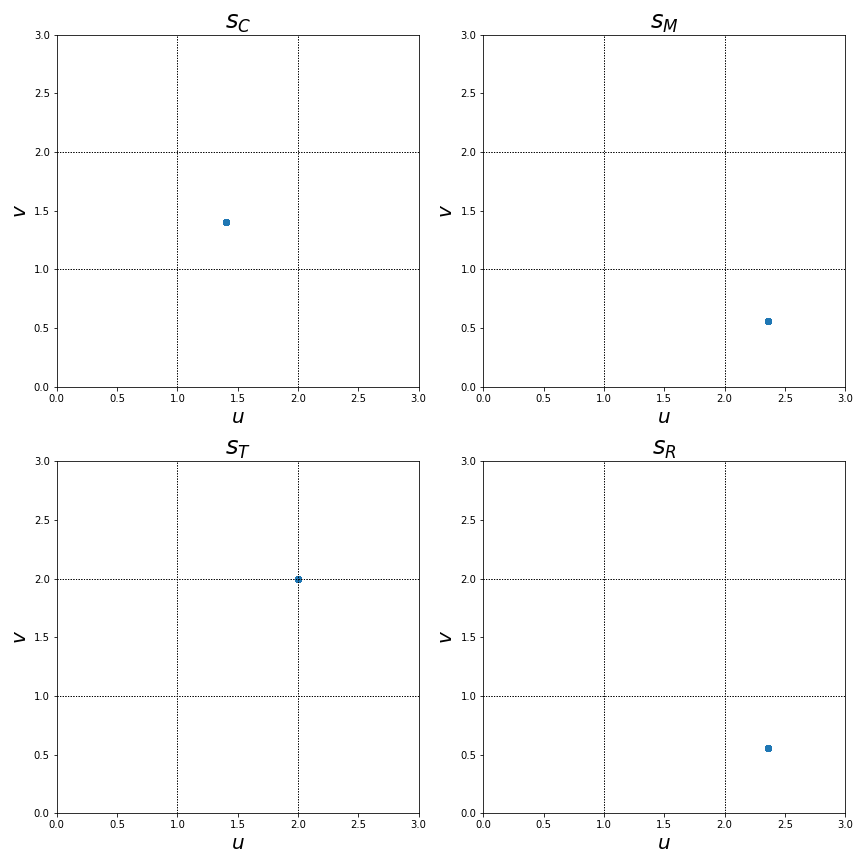}\\(a) $\beta=0.4$\\
   \includegraphics[scale=0.35]{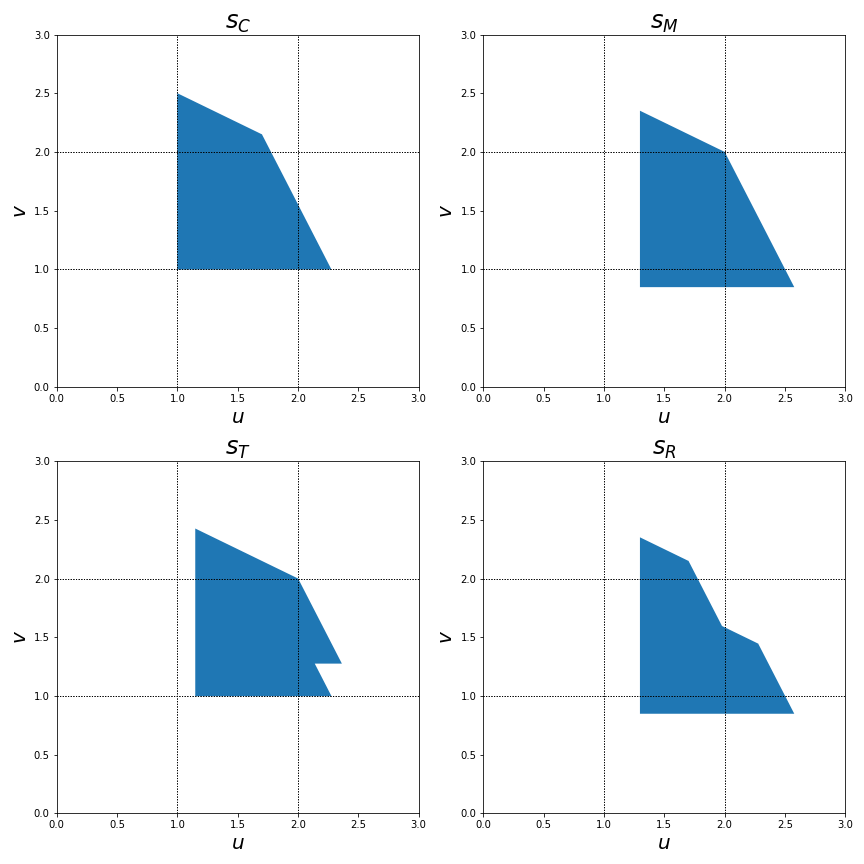}\\(b) $\beta=0.85$
    \caption{\label{fig:SPE_payoffs}Possible continuation payoffs in a subgame perfect equilibrium.}
\end{figure}

Now to construct the sets of possible continuation payoffs $H(s)$
we can start with some arbitrary payoff sets that are guaranteed to
contain $H(s)$, for example (here $\pi_{min}=\min_{p,q}\pi(p,q)$
and $\pi_{max}=\max_{p,q}\pi(p,q)$):
\begin{equation}
H_{0}(s)=[\pi_{min},\pi_{max}]\times[\pi_{min},\pi_{max}]
\end{equation}
And, to calculate the sets $H(s)$ we can use the mapping $F$
from equation (\ref{eq:mapping}) (which maps the set of payoffs
$H=\{H(s)\,for\,all\,s\}$ onto itself) and then construct the sets
of possible SPE payoffs as:
\begin{equation}
\begin{array}{l}
H_{n+1}=F(H_{n})\\
H(s)=\lim_{n\rightarrow\infty}H_{n}(s)
\end{array}
\end{equation}
Note that for any $n$: $H_{n+1}(s)\subset H_{n}(s)$, thus, the
limit exists. Moreover, the limit is exactly $H(s)$: for
any $(u,v)\in\lim_{n\rightarrow\infty}H_{n}(s)$ we can construct
a SPE by recursively calculating the possible continuation payoffs
in the next periods (which always stay within $\lim_{n\rightarrow\infty}H_{n}(s)$
and thus can always be extended further).
\begin{table}
\centering{}%
\begin{tabular}{|c|c|c|}
\hline 
$A\backslash B$ & $p_{M}$ & $p_{C}$\tabularnewline
\hline 
$p_{M}$ & $2,\,2$ & $0,\,3$\tabularnewline
\hline 
$p_{C}$ & $3,\,0$ & $1,\,1$\tabularnewline
\hline 
\end{tabular}\caption{\label{tab:PD_SPE} Payoffs for SPE example.}
\end{table}
\medskip{}

To illustrate the possible equilibria consider the two price setup with payoffs for each customer from Table \ref{tab:PD_SPE}. For given payoffs we can construct the the sets $H(s)$ numerically\footnote{We do that by splitting the initial set $H_{0}(s)$ into a fine grid of squares (10000 by 10000). Then for each $n$ the sets $H_{n+1}$  contain only the squares for which at least one point
$(u,v)$ is a feasible with continuation payoff from $H_{n}$. We
also add a small margin allowing for mistakes of numerical calculations.
Thus, the constructed sets are guaranteed to include $H(s)$ though
could be a bit larger.}. Figure \ref{fig:SPE_payoffs} show the sets of possible payoffs for each algorithms of the opponent for two values of the discount factor $\beta$. We make two points through these simulations. \\

{\it \color{blue} Pure monopoly for low discounting.} If the discount factor is low, then the only possible SPE leads to monopoly outcomes $(p_M,p_M)$. It is always optimal to respond with $s_T$ to $s_T$ or $s_C$. This confirms the result asserted in \citet{bruno_tacit} that as the effective discount factor converges to zero, repeated play of the monopoly outcome emerges as the unique SPE.

In fact our simulations suggest that the threshold on effective discount factor doesn't have to be in the neighborhood of zero. As can be seen in Figure \ref{fig:SPE_payoffs}, this holds for $\beta=0.4$ for any starting strategy. The theoretically interesting fact is that in a standard repeated games setting, low discounting would have led to the uniqueness of $(p_C,p_C)$ as the equilibrium outcome.\\

{\it \color{blue} High discounting leads to large and complex equilibrium payoff sets.} As is widely understood from folk theorems, for high discount factors equilibrium expands to maximal possible size. This continue to hold for the asynchronous repeated we study. Figure \ref{fig:SPE_payoffs} shows the set of equilibrium payoffs for $\beta=0.85$. Note that both competition, i.e. payoff of (1,1), and conflict, i.e. payoff even lower than 1 for one of the sellers, can be sustained. The strategies required to sustain these payoffs though are quite complex. 

Consider for example the competitive equilibrium in which each seller
always chooses $s_{C}$ which would generate continuation payoffs
of $(1,1)$ for the payoff matrix from Table \ref{tab:PD_SPE}. 
Suppose one of the sellers deviates to $s_{T}$ in one of the rounds, and subsequent continuation payoffs are $(u,v)\in H(s_T)$. In order for this deviation to be not profitable, $v$ should be equal to $1$. But note that $u$ is non-trivially different from $1$ ($u\ge1.15$), thus, we should expect a complicated sequence of price pairs and corresponding algorithms that would generate such payoffs. For example, consider the lowest point ($u=1.15$ and $v=1$). For this point we can construct the possible continuation sequences of algorithms that could generate such payoffs. We do this by choosing the algorithm $s'$ such that $(u,v)\in H(s_T,s')$ and then doing this recursively. In the end, we get the following sequence of algorithms ($"\times13"$ means $s_{C}$ is repeated 13 times):
\begin{equation}
s_{T},s_{C}\times13,s_{M},s_{C},s_{M},s_{C}\times2,s_{M}\times10,\dots,\footnote{Here the possible continuation algorithms were always reviewed in the following order: $s_{C}, s_{M}$, $s_{T}$, $s_{R}$. Thus, if there were several
possibilities with which algorithm could continue the earliest one in
the above sequence was always selected. This also means that other
sequences of algorithms are certainly possible. However, all such sequences
have similar irregular structure (we also constructed all possible
sequences up to 15 elements and manually checked them).}
\end{equation}
The above sequence of responses requires some very specific beliefs
about how the opponent should behave in equilibrium while choosing algorithms, which we find
unrealistic. This is the reason why we decided to focus only on Markov perfect
equilibria, which restrict the complexity of the possible strategies.\footnote{In terms of the repeated games literature if we allowed for public randomization, then a lot of the complexity in achieving certain payoffs can be pushed into the randomization devices. It is not clear, however, how to think about public randomization when it come to modeling algorithms.}

\section{A comparison to \citet{bruno_tacit} \label{sec:comparison_bruno}}

The above extension to subgame perfect equilibrium provides a good segue to compare our analysis with \citet{bruno_tacit}. They also look at a duopoly model of price-setting with probabilistic demand and algorithms as strategies, but use subgame perfectness as the equilibrium criterion. The basic idea of using such a set up to explore algorithmic collusion is elegant. We depart from their analysis in several ways.

First, we model the competition in algorithms differently. In \citet{bruno_tacit}, whenever a customer arrives, the price is set by a complex automaton. In our model customer arrivals are unobservable and the algorithms react only to prices set by the competitor. We believe this assumption lends a realism in how algorithms actually interact, especially when we think about online sellers. In addition, it allows us to separate the frequency of customer arrivals from the frequency of possible prices changes. 

Further, we assume that algorithms can react fast to possible prices changes by the competitor, again a reasonable assumption. This allows us to replace the complicated continuous time game with a simpler a asynchronous repeated game in which payoffs in each period are determined by the convergent prices for the current algorithms. This, in turn, allows to derive sharper results, which would be much more difficult to obtain otherwise.

The second key difference is that in \citet{bruno_tacit} the main result about inevitability of monopoly outcomes is stated only for the limiting case in which the frequency of algorithmic adjustments ($\mu$ in our model) converges to $0$.\footnote{\citet{bruno_tacit} also states a result for revision opportunity going to infinity. This means that sellers are changing algorithms faster than consumer arrivals, so before any payoffs accrue, and in the limit even faster than price changes, which is unrealistic.} Then the effective discount factor $\beta=\frac{\mu}{\mu+r}$ also converges to $0$. This limits the dynamic dimension of the strategic interaction between sellers, for they essentially behave "myopically," best responding only to the current algorithm of the opponent, not taking into account how the opponent would react to their choice. 


The case of $\beta\approx 0$ is also unrealistic. Suppose we fix the discount factor at $r=5\%$ per year. Then to generate $\beta=0.5$, which is not even that low, we would need $\mu=0.05$. In plain speak, this means the sellers have the opportunity to revise their algorithms, on average, once in 20 years. If $\beta=0.1$, then it is once in 180 years.  More realistically, suppose we fix $\mu=5$, which  means each seller has on average five opportunities per year to revise their algorithm. Then, $\beta\approx 0.99$. For $\mu=10$, $\beta\approx 0.995$. 

When $\beta$ is high, as we show in Section \ref{subsec:SPE}, the monopoly outcome is no longer inevitable in subgame perfect equilibria, and a wide range of outcomes are possible (as is typical for repeated games). Nevertheless, we are able to show that some collusion is inevitable and extreme collusion is very likely for Markov perfect equilibria. Thus, we can conclude that even for high $\beta$'s, which we argue are more realistic, algorithms will still facilitate collusion.

\section{Final remarks}

\label{final remarks}

The ancient treatise on statecraft, Arthshastra (written and then chiseled between 3rd century BCE and 2nd centruy AD), accredited in folklore to a writer Kautilya, makes several observations about the act of what we term collusive practices, their harm to consumers, and necessity of regulatory practices to reign them in: 
\begin{quote} \begin{small}
[T]raders, joining together and raising or lowering the (prices of) goods,
make a profit of one hundred {\it panas} on one {\it pana} or of hundred {\it kumbhas}
on one {\it kumbh}a. (VIII.4.36, emphasis added) (Pana, a silver coin, was
the basic unit of money, whereas kumbha refers to a measure of weight
(II.29.32))\end{small}
\end{quote}
\begin{quote}\begin{small}
    For traders who by {\it conspiring} together hold back wares or sell them
at a high price, the fine is one thousand {\it panas}. (IV.2.19, emphasis added)\end{small}\footnote{These are cited from \citet{arthsshastra_collusion}.}
\end{quote}
Closer to modern times, the Sherman Act of 1890 was the the first federal legislation in the United States to put down the promotion of free competition an a desired objective of the state. Organized societies has always faced the possibility of extraction of consumer surplus by cartels of producers beyond what is deemed legitimate. But, this line of legitimacy and the means used to cross it without incurring the wrath of the law evolve over time. In the twenty first century, rapid rise in computational technology has arguably equipped firms with a means to collude (almost inadvertently) that is outside the scope of the current antitrust cannon---by delegating pricing to algorithms. 

Algorithmic pricing and its various implications are now a rapidly burgeoning field of enquiry in economics. The purpose of this paper is to offer a simple enough theoretical model that shows why the quickness of response and temporary commitment offered by algorithms lead to inevitability of some collusion and commonality of extreme collusion. The intuition is communicated sharply first through a two-price model and then a generalization to multiple prices shows the basic insights carry through. 

Technically speaking, the analysis here is restricted to two firms and Markov perfect equilibrium. Generalizing to multiple firms and subgame perfect equilibrium with some selection criterion seems like a worthwhile exercise. We also restricted the algorithm to be a function of the current price of the opponent. It could potentially depend on a richer history. In the language of automaton, if the number of internal states of algorithms is not very large (eg. incorporating strategies like grim trigger), then we conjecture that the main results should still go through. This is because even though the initial converge period could be longer, if $\Delta t$ is small enough we can similarly focus on the final price pair. We also discuss in Appendix E why explicitly considering experimentation is easy to handle in our model, since the cost is linear in the set of prices and in unit of time, and hence small. However, if we allow for the aforementioned richer class of algorithms then these experimentation costs can become non-trivial and tackling this trade-off is also an interesting question for future work.


Finally, algorithmic pricing and information technology more generally is now widely recognized to be a challenge for anti-trust regulation. How to incorporate these issues into models of mergers and acquisitions is a fascinating topic for future research.  


\section{Appendix}

The appendix is divided into four parts. Part A contains the proofs missing in the main text. Part B shows that Theorem \ref{prop:multiple_prices} continues to hold even if the "second best" price pair from definition \ref{def:second_best} may not exist, in fact in that case the results is strengthened. Part C presents in detail the extension introduced in Section \ref{sec:precise}---accounting for precise calculations of payoffs which includes trades before prices converge. Part D constructs an unequal equilibrium for an example with five possible prices. Part E discusses the finer details of explicitly modeling experimentation in our framework.  

\medskip

\subsection*{\it \underline{Part A}: Proofs}

\subsection*{Proof of Theorem \ref{th:PD_equilibria}.}
\begin{enumerate}

\item For any seller $i$ and any $s$: $\xi(s,f^{i}(s))\ne(p_{C},p_{M})$.
Otherwise seller $-i$ can keep the strategy $s$, generating the
highest possible payoff $\pi(p_{C},p_{M})$ for them and the lowest
possible $\pi(p_{M},p_{C})$ for their competitor, which cannot be
an equilibrium. Thus, the best response to $s_{C}$ can only
be $s_{C}$ or $s_{T}$ and the the best response to $s_{R}$ cannot
be $s_{M}$.

\item $f^{i}(s_{T})\ne s_{C}$. 
\begin{itemize}
\item Assume the contrary, then the following equilibrium sequences $s_{0},s_{1},s_{2},\dots$
starting with $s_{T}$ are possible:
\begin{equation}
\begin{array}{l}
s_{T},s_{C},s_{T},s_{C},\dots\\
s_{T},s_{C},s_{C},s_{C},\dots\\
s_{T},s_{C},s_{C},s_{T},\dots
\end{array}
\end{equation}
In each case seller $i$ would find it optimal to deviate to $s_{1}=s_{T}$.
\end{itemize}
\item Since $f^{i}(s_{T})\in\{s_{T},s_{M}\}$, we have $\xi(s_{T},f^{i}(s_{T}))=(p_{M},p_{M})$.
Moreover, after we observe a price pair $(p_{M},p_{M})$, all subsequent
prices are $(p_{M},p_{M})$ as well (after $(p_{M},p_{M})$ each could
use $s_{T}$ and guarantee the discounted payoff of at least $(p_{M},p_{M})$
for them.
\item $f^{i}(s_{C})=s_{T}$.
\begin{itemize}
\item Assume the contrary, then $f^{i}(s_{C})=s_{C}$ . Then the following
equilibrium sequences $s_{0},s_{1},s_{2},\dots$ starting with $s_{C}$
are possible
\begin{equation}
\begin{array}{l}
s_{C},s_{C},s_{C},s_{C},\dots\\
s_{C},s_{C},s_{T},\dots
\end{array}
\end{equation}
 In each case seller $i$ would find it optimal to deviate to $s_{1}=s_{T}$.
\end{itemize}
\item If $f^{-i}(s_{R})=s_{C}$ then $f^{i}(s_{R})=s_{C}$.
\begin{itemize}
\item Assume the contrary, then $f^{i}(s_{R})=s_{R}$ and we will observe
the following equilibrium sequence starting with $s_{R}$ for seller
$i$:
\begin{equation}
s_{R},s_{R},s_{C},s_{T},\dots
\end{equation}
For such sequence seller $i$ would find it optimal to deviate to
$s_{1}=s_{C}$.
\end{itemize}
\item $f^{i}(s_{R})=s_{R}$ is possible only if $y<\beta\cdot(x-\beta)$
\begin{itemize}
\item Suppose $f^{i}(s_{R})=s_{R}$ then $f^{-i}(s_{R})=s_{R}$ as well.
\item The expected payoff from responding with $s_{R}$ if your competitor's
algorithm is $s_{R}$:
\begin{equation}
(1-\beta)\cdot\left[\pi(p_{C},p_{M})+\beta\cdot\pi(p_{M},p_{C})+\dots\right]=(1-\beta)\cdot\pi(p_{C},p_{M})+\beta\cdot\frac{\pi(p_{M},p_{C})+\pi(p_{M},p_{C})}{1+\beta}
\end{equation}
The expected payoff from responding with $s_{C}$ if your competitor's
algorithm is $s_{R}$ (the corresponding equilibrium sequence would
be $s_{R},s_{C},s_{T},\dots$):
\begin{equation}
(1-\beta)\cdot\pi(p_{C},p_{M})+\beta(1-\beta)\cdot\pi(p_{C},p_{C})+\beta^{2}\cdot\pi(p_{M},p_{M})
\end{equation}
Thus, the response $s_{R}$ is optimal when
\begin{equation}
\frac{\pi(p_{M},p_{C})+\pi(p_{M},p_{C})}{1+\beta}>(1-\beta)\cdot\pi(p_{C},p_{C})+\beta\cdot\pi(p_{M},p_{M})
\end{equation}
or equivalently when $y<\beta\cdot(x-\beta)$
\end{itemize}
\item $f^{i}(s_{M})=s_{C}$ if $x<\beta$. Otherwise $f^{i}(s_{M})\in\{s_{M},s_{T}\}$.
\begin{itemize}
\item The expected payoff from responding with $s_{C}$ if your competitor's
algorithm is $s_{M}$ (the corresponding equilibrium sequence would
be $s_{M},s_{C},s_{T},\dots$):
\begin{equation}
(1-\beta)\cdot\pi(p_{C},p_{M})+\beta(1-\beta)\cdot\pi(p_{C},p_{C})+\beta^{2}\cdot\pi(p_{M},p_{M})
\end{equation}
\item The expected payoff from responding with $s_{T}$ if your competitor's
algorithm is $s_{M}$ (the corresponding equilibrium sequence would
be $s_{M},s_{T},\dots$):
\begin{equation}
(1-\beta)\cdot\pi(p_{M},p_{M})+\beta(1-\beta)\cdot\pi(p_{M},p_{M})+\beta^{2}\cdot\pi(p_{M},p_{M})
\end{equation}
The first option is preferable if 
\begin{equation}
\pi(p_{C},p_{M})+\beta\cdot\pi(p_{C},p_{C})>(1+\beta)\cdot\pi(p_{M},p_{M})
\end{equation}
or equivalently if $x>\beta$
\item Note, also, that by responding with $s_{M}$ to $s_{M}$ you cannot
get a higher expected payoff than by responding with $s_{T}$ (and
you might only do that the expected payoffs are identical).
\end{itemize}
\end{enumerate}
These points all together prove the result. 

\subsection*{Proof of Corollary \ref{cor:PD_prices}}

We can directly verify that in type I or II equilibria we always end up with monopoly prices. In type III equilibrium if we start with $s_M$ or $s_R$ then the subsequent responses by each seller will always be $s_R$, which would generate the corresponding alternating price sequence. Type III equilibria are only possible if $x>\beta$ and $y<\beta\cdot(x-\beta)$, which is equivalent to the condition stated in the corollary.

\subsection*{Proof of Corollary \ref{cor:PD_monopoly_condition}}

\begin{enumerate}
    \item The first statement directly follows the definitions of $x$ and $y$ from equation (\ref{eq:PD_xy}).
    \item To prove the sufficient condition for ``linear'' demand we can simply calculate the following expression:
    \begin{equation}
    \begin{array}{c}
    \left[\pi(p_{M},p_{M})-\pi(p_{M},p_{M}-x)\right]+\left[\pi(p_{M}-x,p_{M}-x)-\pi(p_{M}-x,p_{M})\right]=\\
    p_{M}\cdot\frac{\alpha\cdot x}{2D}+(p_{M}-x)\frac{-\alpha\cdot x}{2D}=\frac{\alpha x^{2}}{2D}>0
    \end{array}
    \end{equation}
    \item Now consider the ``discrete choice'' profit function. Without loss
    of generality we can assume $b=1$ (we could always scale prices so
    that $b=1$) and replace $a$ with $a\cdot e^{-p_{M}}$:
    \begin{equation}
    \pi(p,q)=\frac{p\cdot e^{-(p-p_{M})}}{a+e^{-(p-p_{M})}+e^{-(q-p_{M})}}
    \end{equation}
    Then
    \begin{equation}
    \pi(p_{M},p_{M})+\pi(p_{M}-x,p_{M}-x)>\pi(p_{M},p_{M}-x)+\pi(p_{M}-x,p_{M})\qquad\Leftrightarrow
    \end{equation}
    \begin{equation}
    \frac{p_{M}}{a+2}+\frac{(p_{M}-x)e^{-x}}{a+2e^{-x}}>\frac{p_{M}+(p_{M}-x)e^{-x}}{a+1+e^{-x}}\label{eq:discrete_choice_inequality}
    \end{equation}
    Finally since for any positive $A,B,C,D$:
    \begin{equation}
    \frac{A}{B}+\frac{C}{D}>2\frac{A+C}{B+D}\qquad\Leftrightarrow\qquad(B-D)\cdot\left(\frac{A}{B}-\frac{C}{D}\right)>0
    \end{equation}
    the inequality (\ref{eq:discrete_choice_inequality}) is equivalent
    to
    \begin{equation}
    2(1-e^{-x})\cdot\left(\frac{p_{M}}{a+2}-\frac{(p_{M}-x)e^{x}}{a+2e^{-x}}\right)>0
    \end{equation}
    which holds since $p_{M}$ is the monopoly price.

\end{enumerate}

\subsection*{Proof of Proposition \ref{prop:before_after}}

\begin{enumerate}
\item First, let us show that all the conditions are necessary for any price
pair $(p',q')$ that follows after $(p,q)$. (This implies that $v^{i}(p,q)$
cannot be above the maximum.)
\begin{itemize}
\item Any algorithm $s$ consistent with $(p,q)$ should map $q\rightarrow p$
which gives us the third condition.
\item In addition, when seller $-i$ gets the opportunity to revise their
algorithm, they can always keep the pair $(p,q)$ (e.g. by not changing
their previous algorithms) and thus $U^{-i}(q',p')$ cannot be below
than $U^{-i}(q,p)$, which gives the second condition. 
\item Finally, since for any $s$: $u^{-i}(s)\ge\underline{u}^{-i}$ then
first condition should also also hold.
\end{itemize}
\item Now suppose $(p',q')$ maximizes $V^{i}(p',q')$ and satisfies all
the necessary conditions. Then $i$ could use the following algorithm
\begin{equation}
s_{*}^{i}(p,q)=\{q\rightarrow p,q'\rightarrow p',x\rightarrow\underline{s}^{i}(x)\,for\,x\ne q\,and\,x\ne q'\}
\end{equation}
to which $-i$ would respond with some  $s^{-i}$ such that $\xi(s_{*}^{i}(p,q),s^{-i})=(p',q')$
with the continuation payoff $u^{-i}(s_{*}^{i}(p,q))=U^{-i}(q',p')$.
Otherwise:
\begin{itemize}
\item If $\xi(s_{*}^{i}(p,q),s^{-i})=(p,q)$ then $u^{-i}(s_{*}^{i}(p,q))\le U^{-i}(q,p)\le U^{-i}(q',p')$
\item If $\xi(s_{*}^{i}(p,q),s^{-i})\ne(p,q)\,or\,(p',q')$ then $u^{-i}(s_{*}^{i}(p,q))\le\underline{u}^{-i}\le U^{-i}(q',p')$
\end{itemize}
\end{enumerate}

\subsection*{Proof of Lemma \ref{lem:uv_pC}}

\begin{enumerate}
\item Each seller when adjusting their algorithm always has an option of
choosing the competitive price $p_{C}$ and keep it forever. Thus
$u^{i}(s_{k})\ge\pi(p_{C},p_{C})$.
\item Consider some $k\ge1$ suppose $\xi(s_{k-1},s_{k})=(q,p)$. Then
\[
u^{-i}(s_{k-1})=(1-\beta)\cdot\pi(p,q)+\beta\cdot v^{-i}(s_{k})\ge\pi(p_C,p_C)
\]
\item If $\pi(p,q)\le\pi(p_{C},p_{C})$ then $v^{-i}(s_{k})\ge\pi(p_{C},p_{C})$. If $\pi(p,q)>\pi(p_{C},p_{C})$ then seller $-i$ can choose $s_{k}'=\{q\rightarrow p,x\rightarrow x\,for\,x\ne q\}$
and keep it forever, with the payoff of at least $\pi(p_C,p_C)$. Thus, $v^{-i}(s_{k})\ge\pi(p_{C},p_{C})$.
\end{enumerate}

\subsection*{Proof of Lemma \ref{lem:uv}}

\begin{enumerate}
\item Note that $v^{-i}(s_{k})=(1-\beta)\cdot\bar{\pi}(s_{k},s_{k+1})+\beta\cdot u^{-i}(s_{k+1})$.
Then, seller $-i$ can keep the strategy $s_{k}$ and generate the
continuation payoff of at least $\bar{\pi}(s_{k},s_{k+1})$, which
proves the first statement of the lemma. 

\item Note that $u^{i}(s_{k})=(1-\beta)\cdot\pi(s_{k+1},s_{k})+\beta\cdot v^{i}(s_{k+1})$.
Suppose $\xi(s_{k},s_{k+1})=(q,p)$.
\item If $q\ne q_{*}^{i}$ then seller $i$ can use the following algorithm\footnote{Here and below we use the following notation to create a new algorithm
$s'$ by slightly adjusting the existing one $s$:
\begin{equation}
s'=\{s|x\rightarrow y\}\qquad\Leftrightarrow\qquad s'(p)=\left\{ \begin{array}{ll}
s(p) & if\,p\ne x\\
y & if\,p=x
\end{array}\right.
\end{equation}
}
\begin{equation}
s^{i}=\{s_{*}^{i}|q\rightarrow p\}
\end{equation}
the continuation payoff for which is either $v_{*}^{i}$ or not lower
than $\pi(p,q)=\pi(s_{k+1},s_{k})$.
\item If $q=q_{*}^{i}$ then seller $i$ can use the following algorithm
\begin{equation}
s^{i}=\{s_{**}^{i}|q\rightarrow p\}
\end{equation}
the continuation payoff for which is either $v_{**}^{i}$ or not lower
than $\pi(p,q)=\pi(s_{k+1},s_{k})$.
\end{enumerate}

\subsection*{Proof of Lemma \ref{lem:constant_payoffs}}

\begin{enumerate}
\item Suppose $u^{i}(s_{k})=v^{i}(s_{k+1})=u^{i}(s_{k+2})=\dots=c$. This
means that each individual payoff of seller $i$ is equal to $c$:
\begin{equation}
\pi(s_{k+1},s_{k})=\bar{\pi}(s_{k+1},s_{k+2})=\pi(s_{k+3},s_{k+2})=\dots=c
\end{equation}
\item Now consider the highest individual payoff of seller $-i$ in the sequence:
\begin{equation}
\bar{\pi}(s_{k},s_{k+1}),\,\pi(s_{k+2},s_{k+1}),\,\bar{\pi}(s_{k+2},s_{k+3}),\,\dots
\end{equation}
Then the corresponding seller can keep the strategy, in response to
which this highest payoff was generated. This would lead to the same
continuation payoff $c$ for seller $i$ and the highest possible
constant continuation payoff for $-i$.

\item Assume the contrary. This implies that $\underline{u}^{-i}<\pi(p_{M},p_{M})$
and $\underline{u}^{i}<\pi(p_{M},p_{M})$. This together with Lemma
A\ref{lem:feasibility_monopoly} implies that an equilibrium with monopoly prices is possible and
after $(p_{M},p_{M})$ all the subsequent prices are $(p_{M},p_{M})$
as well.
\item Suppose for some $k$: $\xi(s_{k},s_{k+1})\ne(p_{M},x)$ for some
$x<p_{M}$ then instead of $s_{k+1}$ the corresponding seller could
use the following strategy
\begin{equation}
s_{k+1}'=\{s_{k+1}|p_{M}\rightarrow p_{M}\}
\end{equation}
which would lead to the same current payoff and monopoly equilibrium
afterwards.
\item If the above $k$ does not exist then for any $k$ for some $x,y<p_{M}$:
\begin{equation}
\xi(s_{k},s_{k+1})=(p_{M},x)\qquad\xi(s_{k+1},s_{k+2})=(p_{M},y)
\end{equation}
Then the corresponding seller instead of $s_{k+2}$ can use the following
strategy
\begin{equation}
s_{k+2}'=\{s_{k+2}|p_{M}\rightarrow p_{M}\}
\end{equation}
 which would generate the current payoff of at least 
\begin{equation}
\pi(s_{k+2}',s_{k+1})\ge\pi(x,p_{M})=\pi(p_{M},y)
\end{equation}
and again the subsequent monopoly equilibrium.
\end{enumerate}

\subsection*{Proof of Lemma \ref{lem:v_SB}}

We prove the lemma by considering possible cases in terms of the first-best prices pairs $(p_*^i,q_*^i)$ from Definition \ref{def:first_best}. Since we need both prices pairs at the same time to simplify exposition we use the following notation: 
\begin{equation}
\label{eq:first_best_consise}
\begin{array}{l}
(p_{*},q_{*})=\arg\max V^{i}(p,q)\qquad s.t.\ U^{-i}(q,p)\ge\underline{u}^{-i}\\
(\dot{q},\dot{p})=\arg\max V^{-i}(q,p)\qquad s.t.\ U^{i}(p,q)\ge\underline{u}^{i}
\end{array}
\end{equation}
Similarly, we will use $(p_{**},q_{**})$ and $(\ddot{q},\ddot{p})$
as the prices corresponding to the second-best maxima from Definition 
\ref{def:second_best}. In general, in this section all prices indexed with letter $p$ (except for $p_{C},p_{C+\Delta},\dots,p_{M-\Delta},p_{M}$)
would correspond to seller $i$ and indexed with letter $q$ would
correspond to seller $-i$.

For the proof we also need Lemmas A\ref{lem:feasibility}, A\ref{lem:feasibility_monopoly}, A\ref{lem:Pareto} stated below and proven after. Note that, if $U^{-i}(q,p)\ge\underline{u}^{-i}$ then seller $i$ can use an algorithm $s = \{q\rightarrow p, x\rightarrow \underline{s}^i(x)\, for\, x\ne q\}$
to generate the expected payoff of at least $\pi(p,q)$, and, thus: 
\begin{equation}
    \begin{array}{l}
        v^i(p,q)\ge v^i(s) \ge V^i(p,q)\ge\pi(p,q) \\
        u^i(p,q)\ge U^i(p,q) = (1-\beta)\cdot\pi(p,q)+\beta\cdot v^i(p,q) \ge \pi(p,q)
    \end{array}
\end{equation}
Similar inequalities holds for seller $-i$ if  $U^{i}(p,q)\ge\underline{u}^{i}$.

\begin{lemmaA}
\label{lem:feasibility}Suppose $\pi(p,q)\ge\pi(p_{C},p_{C})$ and
$\pi(q,p)\ge\pi(p_{C},p_{C})$ and either $(p,q)\ne(\dot{p},q_{*})$ then
either $U^{i}(p,q)\ge\underline{u}^{i}$ or $U^{-i}(q,p)\ge\underline{u}^{-i}$
(or both).
\end{lemmaA}

\begin{lemmaA}
\label{lem:feasibility_monopoly}Either $U^{i}(p_{M},p_{M})\ge\underline{u}^{i}$
or $U^{-i}(p_{M},p_{M})\ge\underline{u}^{-i}$ (or both).
\end{lemmaA}

\begin{lemmaA}
\label{lem:Pareto} If $q_*=p_M$ then for any $p\ne p_*$: either $\pi(p_*,q_*)>\pi(p,q_*)$ or $\pi(q_*,p_*)>\pi(q_*,p)$. Moreover, $V^i(p_{*},q_{*})=\pi(p_{*},q_{*})$ and
$U^{-i}(q_{*},p_{*})=\pi(q_{*},p_{*})$.
\end{lemmaA}

\subsubsection*{Case  $q_{*}\protect\ne p_{M}$ and $\dot{p}\protect\ne p_{M}$}

From Lemma A\ref{lem:feasibility_monopoly} either $U^{i}(p_{M},p_{M})\ge \underline{u}^{i}$ or
$U^{-i}(p_{M},p_{M})\ge \underline{u}^{-i}$ . Suppose for example the latter,
then 
$v_{**}^{i}\ge V^{i}(p_{M},p_{M})\ge\pi(p_{M},p_{M})$.

\subsubsection*{Case $q_{*}=p_{M}$ and $\dot{p}\protect\ne p_{M}$}

\begin{enumerate}
    \item We need to consider only $U^{i}(p_{M},p_{M})<\underline{u}^{i}$. Otherwise, $v_{**}^{-i}\ge V^{-i}(p_{M},p_{M})\ge\pi(p_{M},p_{M})$.

    \item Then, by Lemma A\ref{lem:feasibility_monopoly}: $U^{-i}(p_{M},p_{M})\ge\underline{u}^{-i}$ and, thus,  $\underline{u}^{i}>U^{i}(p_{M},p_{M})\ge\pi(p_{M},p_{M})$.
    
    \item In addition, $p_*\ne p_M$. Otherwise, consider algorithm $s=\{x\rightarrow x\}$. For this algorithm $v^{-i}(s)\ge\pi(p_M,p_M)$ and, thus, $\underline{u}^{i}\le u^i(s)\le \pi(p_M,p_M)$.
    
    \item Let us show that $U^{-i}(p_{M-\Delta},p_{M-\Delta})\ge \underline{u}^{-i}$
    \begin{itemize}
        \item If not, then by Lemma A\ref{lem:feasibility}: $U^{i}(p_{M-\Delta},p_{M-\Delta})\ge\underline{u}^{i}$
        and, thus:
        \begin{equation}
        \begin{array}{ll}
            U^{-i}(p_{M-\Delta},p_{M-\Delta}) & \ge(1-\beta)\cdot\pi(p_{M-\Delta},p_{M-\Delta})+\beta\cdot v^{-i}(p_{M-\Delta},p_{M-\Delta}) \\
            & \ge\pi(p_{M-\Delta},p_{M-\Delta})
        \end{array}
        \end{equation}

        \item However, by Lemma A\ref{lem:Pareto}: $\underline{u}^{-i}\le U^{-i}(q_*,p_*)= \pi(p_{M},p_{*})<\pi(p_{M-\Delta},p_*)\le\pi(p_{M-\Delta},p_{M-\Delta})$.
    \end{itemize}
    
    \item Finally since $U^{-i}(p_{M-\Delta},p_{M-\Delta})\ge \underline{u}^{-i}$:
    \begin{equation}
        \begin{array}{ll}
             v^i_{**}\ge V^{i}(p_{M-\Delta},p_{M-\Delta})&\ge(1-\beta)\cdot\pi(p_{M-\Delta},p_{M-\Delta})+\beta\cdot\underline{u}^{i} \\
             &\ge (1-\beta)\cdot\pi(p_{M-\Delta},p_{M-\Delta})+\beta\cdot\pi(p_{M},p_{M})
        \end{array}
    \end{equation}

\end{enumerate}

\subsubsection*{Case $q_{*}=p_{M}$ and $\dot{p}=p_{M}$}

\begin{enumerate}

\item In this case the monopoly price pair $(p_M, p_M)$ is an absorbing state.

\begin{itemize}
    \item By Lemma A\ref{lem:Pareto}: $\underline{u}^{-i}\le U^{-i}(q_*,p_*)\le \pi(p_M,p_M)$. Similarly, $\underline{u}^i\le \pi(p_M,p_M)$. 
    \item By Lemma A\ref{lem:feasibility_monopoly} for one of the sellers (e.g. $-i$): $U^{-i}(p_M,p_M)\ge \underline{u}^{-i}$, which implies:
        $U^i(p_M,p_M)=(1-\beta)\cdot\pi(p_M,p_M)+\beta\cdot v^i(p_M,p_M)\ge \pi(p_M,p_M)\ge \underline{u}^{i}$.
    \item Finally, $v^{i,-i}(p_M,p_M)\ge\pi(p_M,p_M)$ and $u^{i,-i}(p_M,p_M)\ge U^{i,-i}(p_M,p_M)\ge\pi(p_M,p_M)$, which implies after $(p_M,p_M)$ are subsequent price pairs are also $(p_M,p_M)$.
\end{itemize}

\item Thus, $v^{i,-i}(p_M,p_M) = u^{i,-i}(p_M,p_M) = V^{i,-i}(p_M,p_M)=U^{i,-i}(p_M,p_M) = \pi(p_M,p_M)$.

\item Suppose $U^{i}(p_{M-\Delta},p_{M-\Delta})=(1-\beta)\cdot\pi(p_{M-\Delta},p_{M-\Delta})+\beta\cdot v^{i}(p_{M-\Delta},p_{M-\Delta})>\pi(p_{M},p_{M})$
\begin{itemize}
\item Since $v^{i}(p_{M-\Delta},p_{M-\Delta})\le V^i(p_*,q_*)$ then $p_*\ne p_M$.

\item If $U^{-i}(p_{M-\Delta},p_{M-\Delta})<\underline{u}^{-i}$ then by lemma \ref{lem:feasibility}: $U^{i}(p_{M-\Delta},p_{M-\Delta})\ge\underline{u}^{i}$. 
\begin{itemize}
    \item As a result, $U^{-i}(p_{M-\Delta},p_{M-\Delta})\ge \pi(p_{M-\Delta},p_{M-\Delta})$.
    \item However, by Lemma A\ref{lem:Pareto}:  $\underline{u}^{-i}\le\pi(q_*,p_*)<\pi(p_{M-\Delta},p_{M-\Delta})$.
\end{itemize}

\item Finally, if $U^{-i}(p_{M-\Delta},p_{M-\Delta})\ge\underline{u}^{-i}$:
\begin{equation}
\begin{array}{ll}
V^{i}(p_{M-\Delta},p_{M-\Delta})&=(1-\beta)\cdot\pi(p_{M-\Delta},p_{M-\Delta})+u^{i}(p_{M-\Delta},p_{M-\Delta}) \\
&\ge(1-\beta)\cdot\pi(p_{M-\Delta},p_{M-\Delta})+U^{i}(p_{M-\Delta},p_{M-\Delta}) \\
&\ge(1-\beta)\cdot\pi(p_{M-\Delta},p_{M-\Delta})+\beta\cdot\pi(p_{M},p_{M})
\end{array}
\end{equation}
\end{itemize}

\item Thus, we can assume that $U^{i,-i}(p_{M},p_{M})\ge U^{i,-i}(p_{M-\Delta},p_{M-\Delta})$, which also implies that $v^{i,-i}(p_{M-\Delta},p_{M-\Delta})\ge\pi(p_{M},p_{M})$.

\item Consider $v^{i}(p_{M-\Delta},p_{M-\Delta})=V^{i}(p',q')$:

\begin{itemize}
\item If $q'\ne p_{M}$ then $v_{**}^{i}\ge V^{i}(p',q')\ge\pi(p_{M},p_{M})$.

\item If $q'=p_{M}$ and $p'<p_{M}$ then
\begin{equation}
\begin{array}{l}
V^{i}(p',q')=(1-\beta)\cdot\pi(p',p_{M})+\beta\cdot u^{i}(p',p_{M})\\
U^{-i}(q',p')=(1-\beta)\cdot\pi(p_{M},p')+\beta\cdot v^{-i}(p_{M},p')
\end{array}
\end{equation}
\begin{itemize}
    \item Since $U^{-i}(q',p')\ge U^{-i}(p_{M-\Delta},p_{M-\Delta})\ge(1-\beta)\cdot\pi(p_{M-\Delta},p_{M-\Delta})+\beta\cdot\pi(p_{M},p_{M})$
and $\pi(p_{M-\Delta},p_{M-\Delta})>\pi(p_{M},p')$ then $v^{-i}(p_{M},p')>\pi(p_{M},p_{M})$.
    \item However, by Lemma \ref{lem:uv}: $u^{i}(p',p_{M})\ge V^{i}(p',q')\ge\pi(p_{M},p_{M})$, which implies $v^{-i}(p_{M},p')\le\pi(p_{M},p_{M})$.
\end{itemize}

\item Finally, if $v^{i,-i}(p_{M},p_{M})=V^{i,-i}(p_M,p_M)=\pi(p_M,p_M)$ then
\begin{equation}
    V^{i,-i}(p_{M-\Delta},p_{M-\Delta})=(1-\beta)\cdot\pi(p_{M-\Delta},p_{M-\Delta})+\beta\cdot\pi(p_{M},p_{M})  
\end{equation}
and by Lemma A\ref{lem:feasibility} either $U^{i}(p_{M-\Delta},p_{M-\Delta})\ge\underline{u}^{i}$
or $U^{-i}(p_{M-\Delta},p_{M-\Delta})\ge\underline{u}^{-i}$.
\end{itemize}

\end{enumerate}

Finally to complete the proof of Lemma \ref{lem:v_SB} and hence Theorem \ref{prop:multiple_prices}, we present the proofs of Lemmas A\ref{lem:feasibility}, A\ref{lem:feasibility_monopoly} and A\ref{lem:Pareto}.

\subsection*{Proof of Lemma A\ref{lem:feasibility}}

Suppose $U^{-i}(q,p)<\underline{u}^{-i}$ and $q\ne q_{*}$. If we consider
algorithm $s^{i}=\{s_{*}^{i}|q\rightarrow p\}$ for seller $i$ then $v^{i}(s^{i})=v_{*}^{i}$, which implies that:
\begin{equation}
U^{i}(p,q)=(1-\beta)\cdot\pi(p,q)+\beta\cdot v^{i}(p,q)\ge(1-\beta)\cdot\pi(p_{C},p_{C})+\beta\cdot v_{*}^{i}\ge\underline{u}^{i}.
\end{equation}

\subsection*{Proof of Lemma A\ref{lem:feasibility_monopoly}}
\begin{enumerate}

\item By Lemma A\ref{lem:feasibility} we only need to consider the
case in which $q_{*}=\dot{p}=p_{M}$.

\item Assume $U^{i}(p_{M},p_{M})<\underline{u}^{i}$ and $U^{-i}(p_{M},p_{M})<\underline{u}^{-i}$
then in any equilibrium sequence we should never observe a monopoly
price pair $(p_{M},p_{M})$.

\item Consider some arbitrary equilibrium sequence $s_{0},s_{1},s_{2},...$
Suppose ``odd'' algorithms correspond, for example, to seller $i$
and, in addition:
\begin{equation}
\begin{array}{l}
\xi(s_{0},s_{1})=(x,p_{M})\qquad x\ne p_M\\
\xi(s_{1},s_{2})=(y,z)\qquad z\ne p_M
\end{array}
\end{equation}
then the monopoly price pair $(p_{M},p_{M})$ is possible.

\begin{itemize}
\item To prove this first note that:
\begin{equation}
\underline{u}^{i}\le u^{i}(s_{0})=(1-\beta)\cdot\pi(p_{M},x)+\beta\cdot V^{i}(y,z)<(1-\beta)\cdot\pi(p_{M},p_{M})+\beta\cdot v_{**}^{i}
\end{equation}
\item Now consider the following two unrelated algorithms:
\begin{equation}
\begin{array}{l}
s^{-i}=\{\underline{s}^{-i}|p_{M}\rightarrow p_{M}\}\\
s^{i}=\{s_{**}^{i}|p_{M}\rightarrow p_{M}\}
\end{array}
\end{equation}
\item By responding with $s^{i}$ to $s^{-i}$ the seller $i$ would get
at least $(1-\beta)\cdot\pi(p_{M},p_{M})+\beta\cdot v_{**}^{i}>\underline{u}^{i}$
while if the next price pair is not $(p_{M},p_{M})$ they would get
at most $\underline{u}^{i}$.
\end{itemize}

\item Consider an equilibrium sequence $s_{0},s_{1},s_{2},\dots$ starting
with $s_{0}=s_{T}=\{x\rightarrow x\,for\,any\,x\}$. In this sequence:
\begin{itemize}
\item $\xi(s_{0},s_{1})=(x_{0},x_{0})\ne(p_{M},p_{M})$.
\item $\xi(s_{1},s_{2})=(x_{1},p_{M})$: otherwise we can replace $s_{1}$
with $s_{1}'=\{s_{1}|p_{M}\rightarrow p_{M}\}$ and improve payoffs
of both sellers in the previous period.
\item $\xi(s_{k},s_{k+1})=(x_{k},p_{M})$ for $k>1$: otherwise by previous
conjecture there is should exist a monopoly price pair.
\item But, then this sequence cannot be an equilibrium. The seller with
the highest $\pi(x_{k},p_{M})$ would find optimal to keep the algorithm
$s_{k}$.
\end{itemize}
\end{enumerate}

\subsection*{Proof of Lemma A\ref{lem:Pareto}}
\begin{enumerate}

\item Consider $V^{i}(p_{*},q_{*})$ and the sequence of price pairs that follows
$(p_{*},q_{*})$ in equilibrium:
\begin{equation}
(p_{*},q_{*})\rightarrow(p',q')\rightarrow(p'',q'')\rightarrow\dots
\end{equation}

\item By Lemma \ref{lem:uv}: $V^i(p_*,q_*)\le U^i(p',q') $ and $V^{-i}(q', p')\le U^{-i}(q'',p'')$. 

\item In addition for $q_{*}=p_{M}$: $U^{-i}(q_*,p_*)\le V^{-i}(q', p')$
\begin{itemize}
    \item $U^{i}(p_{*},q_{*})=(1-\beta)\cdot\pi(p_{*},q_{*})+\beta\cdot v_{*}^{i}\ge \underline{u}^{i}$ since $\underline{u}^{i}\le u^{i}(\{x\rightarrow p_c\})\le (1-\beta)\cdot\pi(p_C,p_C)+\beta\cdot v_{*}^{i}$.
    \item As a result, $V^{-i}(q',p')\ge V^{-i}(q_*,p_*) =(1-\beta)\cdot\pi(q_{*},p_{*})+\beta\cdot u^{-i}(q_*,p_*)\ge (1-\beta)\cdot\pi(q_{*},p_{*})+\beta\cdot U^{-i}(q_*,p_*)$.
    \item Finally, since $U^{-i}(q_*,p_*)=(1-\beta)\cdot \pi(q_*,p_*)+\beta\cdot V^{-i}(q',p')$ then $V^{-i}(q', p')\ge\pi(q_*,p_*)$ and, thus, $V^{-i}(q', p')\ge U^{-i}(q_*,p_*)$.

\end{itemize}

\item If $U^{i}(p',q')\le V^{i}(p'',q'')$ then since $(p_{*},q_{*})$
is the first-best price pair for $i$:
\begin{equation}
\begin{array}{l}
V(p_{*},q_{*})=U(p',q',)=V(p'',q'')=\pi(p_{*},q_{*})=\pi(p',q') \\
U(q_{*},p_{*})=U(q',p')=U(q'',p'')=\pi(q_{*},p_{*})=\pi(q',p')
\end{array}
\end{equation}

\begin{itemize}
\item Suppose, there exists $\tilde{p}$ such that $\pi(\tilde{p},q_{*})\ge\pi(p_*,q_*)$ and $\pi(q_{*},\tilde{p})>\pi(q_*, p_*)$.

\item Consider 
$U^{-i}(q_{*},\tilde{p})=(1-\beta)\cdot\pi(q_{*},\tilde{p})+\beta\cdot v^{-i}(q_{*},\tilde{p})$. Since $-i$ could use the following algorithm: 
\begin{equation}
s^{-i}=\{p_{*}\rightarrow q_{*},\tilde{p}\rightarrow q_{*},x\rightarrow\underline{s}^{-i}(x)\,for\,x\ne p_{*},\tilde{p}\}
\end{equation}
 then $v^{-i}(q_{*},\tilde{p})\ge v^{-i}(s^{-i})\ge\pi(q_{*},p_{*})\ge\underline{u}^{-i}$
and, thus, $U^{-i}(q_{*},\tilde{p})\ge\underline{u}^{-i}$. But, then
$(p_{*},q_{*})$ is not the first best price pair.
\end{itemize}

\item If $U^{i}(p',q')>V^{i}(p'',q'')$ the by lemma \ref{lem:uv} either $(p'',q'')=(p_*,q_*)$ or
$(p'',q'')=(p_{**},q_{**})$ and $q'=q_*$.

\item If $(p'',q'')=(p_{*},q_{*})$
then
\begin{equation}
\begin{array}{l}
U^{-i}(q_{*},p_{*})=U^{-i}(q',p')=\pi(q_{*},p_{*})=\pi(q',p')\\
\pi(p',q')>\pi(p_{*},q_{*})
\end{array}
\end{equation}
However, then $V^{i}(p',q')>V^{i}(p_{*},q_{*})$ and $U^{-i}(q',p')\ge U^{-i}(q_{*},p_{*})\ge\underline{u}^{-i}$
and, thus $(p_{*},q_{*})$ is not the first best price pair.

\item If $(p'',q'')=(p_{**},q_{'**})$ and $q'=q_{*}$ then:
\begin{itemize}

\item $U^{i}(p',q')>V^{i}(p'',q'')$ implies that $\pi(p_{*},q_{*})<\pi(p',q_{*})$.

\item If $\pi(q_{*},p')>\pi(q_{*},p_{*})$ then
\begin{equation}
\begin{array}{l}
U^{-i}(q_{*},p')\ge(1-\beta)\cdot\pi(q_{*},p')+\beta\cdot V(q_{*},p')>U^{-i}(q_{*},p_{*})\ge\underline{u}^{-i}\\
V^{i}(p',q_{*})\ge(1-\beta)\cdot\pi(q_{*},p')+\beta\cdot U(p',q_{*})>V^{i}(p_{*},q_{*})
\end{array}
\end{equation}
and, thus, again $(p_{*},q_{*})$ is not the first best price pair.
\end{itemize}

\item Finally, let us prove that $V^i(p_{*},q_{*})=\pi(p_{*},q_{*})$ and
$U^{-i}(q_{*},p_{*})=\pi(q_{*},p_{*})$.

\begin{itemize}
    \item Since $U^{i}(p_{*},q_{*})\ge\underline{u}^{i}$ then $V^{-i}(q',p')\ge V^{-i}(q_*,p_*)\ge\pi(q_*,p_*)$ which in turns implies that $U^{-i}(q_{*},p_{*})\ge\pi(q_{*},p_{*})$. In addition, $V^i(p_*,q_*)\ge \pi(p_*,q_*)$.
    \item By assumption 0--(iv) from section \ref{subsec:market}: $\left(\pi(p_*,q_*), \pi(q_*,p_*)\right)$ lies on the convex Hull of the payoffs set\footnote{Even if $p_*<\arg\max\tilde{p} \pi(\tilde{p},p_M)$ then $(p_*,q_*)$ is the point with the highest possible payoff on one seller.}, which means the above inequalities can only hold if all subsequent price pairs are $(p_*,q_*)$.
\end{itemize}
\end{enumerate}

\subsection*{Proof of Lemma \ref{lem:sTsR}}
\begin{enumerate}
\item Similarly, as in the proof of Proposition \ref{sec:PD} we can show it is never optimal to respond with an algorithm that would generate an immediate payoff of $\pi(p_M,p_C)$ and that $f^i(s_T)\ne s_C$.
\item Consider some equilibrium sequence: $s_{0},s_{1},s_{2},s_{3},...$
\item The following two equilibrium sequences are not possible:
\begin{itemize}
\item $s_{T},s_{R},s_{C},s_{3},\dots$: ``odd'' seller can deviate to
$s_{1}'=s_{T}$ and $s_{3}'=s_{3}$ (note that $s_{3}$is either $s_{C}$
or $s_{T}$) and improve their payoff
\item $s_{R},s_{R},s_{C},s_{3},\dots$: ``odd'' seller can deviate to
$s_{1}'=s_{C}$ and $s_{3}'=s_{3}$
\end{itemize}
\item Suppose $f^{i}(s_{T})=s_{R}$ then given the above only the following
equilibrium sequences starting with $s_{T}$ chosen by seller $-i$:
\begin{itemize}
\item $s_{T},s_{R},s_{T},s_{R},...$: ``odd'' seller can deviate to $s_{1}'=s_{T}$
and $s_{3}'=s_{3}=s_{R}$
\item $s_{T},s_{R},s_{R},s_{R},...$: ``odd'' seller can deviate to $s_{1}'=s_{T}$
and then keep this algorithm
\item $s_{T},s_{R},s_{R},s_{T},s_{R},s_{T},s_{R},...$: same as above starting
from $s_{3}$
\item $s_{T},s_{R},s_{R},s_{T},s_{T},...$ or $s_{T},s_{R},s_{R},s_{T},s_{M},...$:
``odd'' seller can deviate to $s_{1}'=s_{T}$
\end{itemize}
\end{enumerate}

\subsection*{Proof of Proposition \ref{prop:PD_equilibria_cycles}}
\begin{enumerate}
\item Because of Lemma \ref{lem:sTsR} the first five items in the proof of proposition
\ref{th:PD_equilibria} are still valid.
\item Consider one seller and suppose they are responding to algorithm $s_{R}$
chosen by the competitor.
\begin{itemize}
\item If the seller chooses $s_{C}$ then the expected payoff is :
\begin{equation}
(1-\beta)\cdot\pi(s_{C},s_{R})+\beta(1-\beta)\cdot\pi(s_{C},s_{T})+\beta^{2}\cdot\pi(p_{M},p_{M})=(1-\beta)\cdot(1+x)+\beta^{2}
\end{equation}
\item If the seller chooses $s_{T}$ then the expected payoff is: 
\begin{equation}
(1-\beta)\cdot\pi(s_{T},s_{R})+\beta\cdot\pi(p_{M},p_{M})
\end{equation}
which is equal to (for two options of resolving cycles):
\begin{equation}
\begin{array}{l}
(minimum\,price)\qquad\,\,\,\, =\beta\\
(average\,payoff)\qquad =(1-\beta)\cdot\left(\frac{1}{2}+\frac{x-y}{4}\right)+\beta
\end{array}
\end{equation}
\item Responding with $s_{T}$ could be preferable to $s_{C}$ only for
``average payoff'' case and only if
\begin{equation}
(1-\beta)\cdot(1+x)+\beta^{2}\le(1-\beta)\cdot\left(\frac{1}{2}+\frac{x-y}{4}\right)+\beta\qquad\Leftrightarrow
\end{equation}
\begin{equation}
y\le4\beta-2-3x
\end{equation}
which is only possible if $\beta\ge\frac{1}{2}$ and $x<\frac{4\beta-2}{3}<\beta$
(thus, only type I equilibrium could be affected).
\end{itemize}
\end{enumerate}

\medskip

\subsection*{\it \underline{Part B}: Non-existent "second best" price pair}

In this section we consider a special case, in which the "second best" price pair from Definition \ref{def:second_best} does not exist for at least one of the sellers. We can show that in this case for an arbitrary equilibrium sequence $s_0, s_1, s_2, \dots$ the continuation payoffs for at least one of the seller will exceed the monopoly payoff $\pi(p_M, p_M)$. The steps bellow prove this result. For the "first best" pair pairs we use the same notation equation (\ref{eq:first_best_consise}) as in the proof of Lemma \ref{lem:v_SB}.

\begin{enumerate}

\item Suppose, for example, $U^{-i}(q,p)<\underline{u}^{-i}$ for any $p, q\ne q_*$. 
This implies $U^{-i}(q_*,p)\ge\underline{u}^{-i}$ for any $p$ (seller $i$ can construct an algorithm which would lead to prices $(q_*,p)$).

\item As a result, Lemma \ref{lem:uv} for seller $i$ means either $v^i(s_{k+1})\ge u^i(s_k)$ or $v_i(s_{k+1})=v^i_*$. Thus, the continuation payoffs for seller $i$ in the any equilibrium sequence either become constant (the we can use Lemma \ref{lem:constant_payoffs}) or reaches $v_i^*$. 

\item If $q_*=p_M$ then $v^i_*\ge V^i(p_M,p_M)\ge\pi(p_M,p_M)$. Thus, assume $q_*\ne p_M$.

\item If $\dot{p}\ne p_M$ then $v^{-i}_{**}\ge V^{-i}(p_M,p_M)\ge\pi(p_M,p_M)$. Thus, assume $\dot{p}=p_M$. (Since $U^i(p_M,p_M)\ge\pi(p_M,p_M)$, the "second best" price pair necessarily exists for seller $-i$).

\item Since $U^{-i}(\dot{q},\dot{p})=(1-\beta)\cdot\pi(\dot{q},\dot{p})+\beta\cdot v_{*}^{-i}\ge u^{-i}(\{x\rightarrow p_C\})\ge\underline{u}^{-i}$ then $\dot{q}=q_*$.

\item Finally for any $p$: $U^{-i}(q_*,p)=(1-\beta)\cdot\pi(q_*,p)+\beta\cdot v^{-i}(q_*,p)$. Then, either $p=\dot{p}$ and $v^{-i}(q_*,p)=v^{-i}_*$ or $p\ne\dot{p}$. This together with Lemma \ref{lem:uv} implies that in any equilibrium sequence either continuation payoffs for seller $-i$ converge to a constant or they become above $v^{-i}_*\ge V^{-i}(p_M,p_M)\ge\pi(p_M,p_M)$. 

\end{enumerate}
\medskip

\subsection*{\it \underline{Part C}: Explicitly including pre-convergent prices into the calculations}

Here we define the precise model and calculate expected payoffs within this model, and study the Markov equilibria in this updated model. A key difference is that the payoff relevant variable now is not only the algorithm of the opponent, but also the next price that this algorithm chooses. Then, we formally establish and prove the claim made in Section \ref{sec:precise}. For simplicity and to relax assumptions piecemeal, we continue to assume that the seller with the revision opportunity deciphers the opponent's algorithm at a very small cost, see part E of the Appendix for a justification.

\subsubsection*{Precise Model \label{subsec:precise_payoffs}}

Similar to our main model we assume
that prices change in the discrete moments $0,1\cdot\Delta t,2\cdot\Delta t,\dots$
with subsequent prices determined by
\begin{equation}
\begin{array}{l}
p_{t+\Delta t}^{A}=s^{A}(p_{t}^{B})\\
p_{t+\Delta t}^{B}=s^{B}(p_{t}^{A})
\end{array}
\end{equation}
where $s^{A}$ and $s^{B}$ are the current algorithms of the two
sellers. Customers arrive continuously and randomly (also potentially between price adjustments) with Poisson intensity $\lambda$. The continuously compounded discount
rate is $r$. Before each price adjustment the seller who was not
the last to adjust their algorithm gets the opportunity to adjust
their algorithms with probability $(1-e^{-\mu\Delta t})$. When this
happens this seller can also choose prices for the current and
for the next period. 

As before, we {\it index} all moments in time when one of the sellers 
adjusting their algorithm with $k=0,1,2,\dots$. and consider some
arbitrary adjustment moment $k$ that happens at time $t_{k}$. We
express the ``precise'' expected payoff of the seller adjusting
their algorithm, $\tilde{U}_{k}$, and of the other seller, $\tilde{V}_{k}$,
as follows (note that $\int_{0}^{\Delta t}\lambda\pi_{k,i}e^{-rt}dt=\frac{\lambda}{r}(1-e^{-r\Delta t})\pi_{k,i}$):
\begin{equation}
\begin{array}{l}
\tilde{U}_{k}=\sum_{i=0}^{\infty}e^{-(r+\mu)i\Delta t}\cdot\left[\frac{\lambda}{r}(1-e^{-r\Delta t})\pi_{k,i}+(1-e^{-\mu\Delta t})\tilde{V}_{k+1}\right]\\
\ \\
\tilde{V}_{k}=\sum_{i=0}^{\infty}e^{-(r+\mu)i\Delta t}\cdot\left[\frac{\lambda}{r}(1-e^{-r\Delta t})\bar{\pi}_{k,i}+(1-e^{-\mu\Delta t})\tilde{U}_{k+1}\right]
\end{array}
\end{equation}
where $\pi_{k,i}$ and $\bar{\pi}_{k,i}$ are the expected payoff
from one customer for the corresponding seller in the interval $(t_{k+i\cdot\Delta t},t_{k+(i+1)\cdot\Delta t})$,
in which the prices are constant.

If we denote average weighted per period payoffs of each sellers as:
\begin{equation}
\pi_{k}=\frac{\sum_{i=0}^{\infty}e^{-(r+\mu)i\Delta t}\pi_{k,i}}{\sum_{i=0}^{\infty}e^{-(r+\mu)i\Delta t}}\qquad\bar{\pi}_{k}=\frac{\sum_{i=0}^{\infty}e^{-(r+\mu)i\Delta t}\bar{\pi}_{k,i}}{\sum_{i=0}^{\infty}e^{-(r+\mu)i\Delta t}}
\label{eq:pi_average}
\end{equation}
and normalize payoffs as $\tilde{u}_{k}=\frac{\tilde{U}_{k}}{\lambda/r}\cdot e^{-\mu\Delta t}$
and $\tilde{v}_{k}=\frac{\tilde{V}_{k}}{\lambda/r}\cdot e^{-\mu\Delta t}$
then we get the following equations describing the ``precise'' normalized
payoffs:
\begin{equation}
\begin{array}{l}
\tilde{u}_{k}=(1-\tilde{\beta})\cdot\pi_{k}+\tilde{\beta}\cdot\tilde{v}_{k+1}\\
\tilde{v}_{k}=(1-\tilde{\beta})\cdot\bar{\pi}_{k}+\tilde{\beta}\cdot\tilde{u}_{k+1}
\end{array}\label{eq:uv_precise}
\end{equation}
Here $\tilde{\beta}=\frac{1-e^{-\mu\Delta t}}{1-e^{-(\mu+r)\Delta t}}$
is the effective discount factor.

Equations (\ref{eq:uv_precise}) are similar to equations (\ref{eq:uv})
describing the evolution of normalized payoffs in the main model. In
addition, $\tilde{\beta}\rightarrow\beta=\frac{\mu}{\mu+r}$ as $\Delta t\rightarrow0$. Finally, if  $\pi(s_{k},s_{k-1})$ and $\bar{\pi}(s_{k-1},s_{k})$ are the limiting expected payoffs from one customer for the two algorithms $s_{k}$ and $s_{k-1}$\footnote{To be precise these payoffs also depend on the initial price sequence chosen by the seller adjusting their algorithms. We skip this here for more concise notation.}, then it follows directly from equations (\ref{eq:pi_average}) that:
\begin{equation}
\begin{array}{l}
\left|\pi_{k}-\pi(s_{k},s_{k-1})\right|\le\Delta\pi_{max}\cdot\left(1-e^{-(r+\mu)\cdot N_{k}\Delta t}\right)\\
\left|\bar{\pi}_{k}-\bar{\pi}(s_{k-1},s_{k})\right|\le\Delta\pi_{max}\cdot\left(1-e^{-(r+\mu)\cdot N_{K}\Delta t}\right)
\end{array}
\end{equation}
where $\Delta\pi_{max}=max_{p,q,p',q'}\left|\pi(p,q)-\pi(p',q')\right|$
and $N_k$ is the length of the converge period for the current pair of algorithms (which cannot be higher than the number of possible prices squared). Thus, for the same behavior of the sellers the expected payoffs in the ``precise'' model (\ref{eq:uv_precise}) would converge to the expected payoffs of the ``approximate'' limiting model (\ref{eq:uv}) as $\Delta t \rightarrow 0$.\footnote{We again assume sellers cannot respond with an algorithm that would generate a cycle. This assumption can be easily relaxed though. If we resolve cycles in the main model as equivalent to obtaining the average payoffs over the cycle, then similarly the ``precise'' payoffs $\pi_{k}$and $\bar{\pi}_{k}$would converge to these averages as $\Delta t\rightarrow0$.}

\subsubsection*{Equilibria in the Precise Model \label{subsec:precise_equilibria}}

For sufficiently small $\Delta t$ the initial price would not have
any significant effect on the expected payoff of the seller adjusting
their algorithm---by choosing the appropriate initial prices they can
always get the continuation sequence with the highest expected payoffs.
However, it could be used as a coordination device allowing the seller
to select one of the possible best responses with practically identical payoffs.
In order not to do deal with such situations we assume below that
all payoffs in the payoff matrix $\pi(p,q)$ are distinct, that is: 
\begin{equation}
\pi(p,q)=\pi(p',q')\ \ \Leftrightarrow\ \ (p,q)=(p',q')\label{eq:distinct_prices}
\end{equation}

Now consider some arbitrary Markov best response functions for
two sellers:
\begin{equation}
\{f^{i}(s,p),\ p_{0}^{i}(s,p),\ p_{1}^{i}(s,p)\}\qquad i=A,B\label{eq:BR_pair}
\end{equation}
Here $f^{i}(s,p)$ is the algorithm with which seller $i$ would respond
and $p_{0}^{i}(s,p)$ and $p_{1}^{i}(s,p)$ are the initial two prices
chosen. Note, the state variable here includes, in addition to the current algorithm of the opponent $s$, the first price $p$ that the opponent's algorithms will choose. For the pair of responses (\ref{eq:BR_pair}) we can calculate
the expected payoffs of the seller adjusting their algorithm $\tilde{u}^{i}(s,p)$
and of the other seller $\tilde{v}^{i}(s,p)$ for given values of $s$ and $p$. Also, denote by $u^{i}(s,p)$ and $v^{i}(s,p)$
the limits of these expected payoffs as $\Delta t\rightarrow0$ (for
the same response functions of the sellers (\ref{eq:BR_pair})):
\begin{equation}
\begin{array}{l}
u^{i}(s,p)=\lim_{\Delta t\rightarrow0}\tilde{u}(s,p)\\
v^{i}(s,p)=\lim_{\Delta t\rightarrow0}\tilde{v}(s,p)
\end{array}
\end{equation}

\begin{proposition}
\label{prop:precise}Suppose that for any $\varepsilon>0$ there exists
$\Delta t=\Delta t_{\varepsilon}<\varepsilon$ such that the pair
of responses from (\ref{eq:BR_pair}) is an equilibrium in the ``precise
model'' then for any $\beta$ (except for possibly some finite
set of values):
\begin{enumerate}
\item $u^{i}(s,p)=u^{i}(s,p')$ and $v^{i}(s,p)=v^{i}(s,p')$ for any prices
$p$ and $p'$
\item For arbitrary $p^{i}(s)$ the strategies 
\begin{equation}
\tilde{f}^{i}(s)=f^{i}(s,p^{i}(s))\label{eq:price_dependence}
\end{equation}
 are an equilibria in the ``main model" with expected
payoffs $u^{i}(s)=u^{i}(s,p')$ and $v^{i}(s)=v^{i}(s,p')$ for any
$p'$.
\end{enumerate}
\end{proposition}

Proposition \ref{prop:precise} shows that if some pair of best
responses (\ref{eq:BR_pair}) survives as an equilibrium for an arbitrary
small $\Delta t$ then we can construct an equilibrium in the
main model by resolving the possible dependence
on price arbitrarily (equation (\ref{eq:price_dependence})). Moreover,
the expected payoffs in the limiting main model would
be equal to the limits of the expected payoffs of the precise
model as $\Delta t\rightarrow0$. Since there is a finite number
of possible response pairs (\ref{eq:BR_pair}) the next corollary
directly follows.
\begin{corollary}
For any $\beta=\frac{\mu}{\mu+r}$ (except for possibly some finite
set) and for sufficiently small $\Delta t$ if we take any Markov
equilibrium in the ``precise'' model and resolve dependency on the
initial price arbitrarily then the obtained strategies would
be an equilibrium in the ``approximate'' limiting model. 
\end{corollary}
\noindent Thus, if we observe only collusive equilibria in the main
model, then all equilibria in the ``precise model'' for sufficiently small
$\Delta t$ would also be collusive (except for possibly some non-generic
subset of discount factors $\beta$).

\subsection*{Proof of Proposition \ref{prop:precise}}

\begin{enumerate}
\item Consider one of the sellers choosing the best response for some $s$
and $p$. For sufficiently small $\Delta t$ the seller would always
respond with an algorithm $s'$ and initial price sequence that would
generate the highest possible limiting expected payoff:
\begin{equation}
u(s,p)=\max_{p',s'(.)}(1-\beta)\cdot\pi(p',s(p'))+\beta\cdot v(s',p')\qquad s.t.\ s'(s(p'))=p'
\end{equation}
which does not depend on $p$.
\item If if there are several different $p'$ and $s'$ that generate the
highest limiting payoff for the seller adjusting their algorithm then
with assumption (\ref{eq:distinct_prices}) except for possibly finite
set of discount factors $\beta$ these responses also generate exactly
the same limiting expected payoff for the other seller:
\begin{equation}
v(s,p)=(1-\beta)\cdot\pi(s(p'),p')+\beta\cdot u(s',p')
\end{equation}
\item As in the main model, for sufficiently small $\Delta t$ the seller
adjusting the algorithm would also choose the initial prices that
maximizes the current expected payoff (since any future responses
of the opponent generate the same limiting expected payoff irrespective
of the future initial prices).
\item Finally, the strategies (\ref{eq:price_dependence}) are an equilibrium
in the ``approximate'' limiting model since any profitable deviation
would also be profitable in the ``precise'' model for sufficiently
small $\Delta t$.
\end{enumerate}

\medskip

\subsection*{\it \underline{Part D}: Example with unequal equilibrium}

In this section we construct an example of an unequal equilibrium, in which prices can converge to a collusive that is not $(p_M,p_M)$. For this example we will use the discrete choice demand function with the following expected payoff from one customer:
\begin{equation}
\pi(p,q)=p\cdot\frac{e^{-bp}}{a+e^{-bp}+e^{-bq}}
\end{equation}
We choose parameters $a\approx0.0158$ and $b\approx0.4760$ so that the monopoly price is $p_{M}=8$ and the competitive price is $p_{C}=4$. Each seller can chose any prices from the set $\{4,5,6,7,8\}$. The payoff matrix for these prices with the above parameters is shown in Table \ref{tab:UE_payoffs}.

To describe an equilibrium instead of specifying all algorithms that
are used (which would be quite complicated) we instead specify the transition
matrices, which show which price pair could follow another price pair.
Later we can verify that the chosen transition matrices can indeed be sustained
in an equilibrium. More specifically, consider the best algorithm that could follow prices
$(p,q)$ for seller $i$ (from equation (19)):
\begin{equation}
s^{i}(p,q)=\arg\max_{s:q\rightarrow p}v^{i}(s)
\end{equation}
Now denote by $\phi^{i}(p,q)$ the price pair that would follow $(p,q)$
on the equilibrium sequence starting with $s^{i}(p,q)$ or more precisely
\begin{equation}
\phi^{i}(p,q)=\xi^{-i}(s^{i}(p,q),f^{-i}(s^{i}(p,q))
\end{equation}
If we know the transition matrices $\phi^{i}(p,q)$
for each seller we can completely describe how prices will evolve and
we can calculate the expected payoffs matrices $v^{i}(p,q)$,
$u^{i}(p,q)$, $V^{i}(p,q)$, $U^{i}(p,q)$. Finally, using Proposition
\ref{prop:before_after} we verify that the specified transitions are feasible and optimal, and, thus, they indeed correspond to an equilibrium. The specific
algorithms necessary to sustain such equilibrium can then be recovered from
equation (\ref{eq:transition_algorithm}).

Consider the transition matrices specified in Table \ref{tab:UE_transition}. We verified that the specified transition matrices can be sustained in  an equilibrium for $\beta=0.9,\,0.95\,,0.999$ and, also, for any $\beta$ sufficiently close $1$.\footnote{We performed the calculations numerically with a Python script. Nevertheless, when comparing two expected payoffs we required for the difference to be higher than the specified threshold, which was much higher than possible mistakes in the numerical calculations. Thus, we believe
the proof is strict. To show that the transition matrices could be
sustained as an equilibrium for $\beta$ sufficiently close to $1$
we created a slightly modified version of the script, which compared
not the discounted expected payoffs but the sequences of payoffs.}
For the specified transition matrices, depending on the initial condition we end up either with monopoly prices $(8,8)$ or with an unequal outcome $(7,8)$ in which seller $A$ gets a payoff above the monopoly. In order to sustain prices $(7,8)$ in an equilibrium seller $A$ can use the following algorithm:
\begin{equation}
s^A = \{8\rightarrow 7,\ 7\rightarrow 6,\ 6\rightarrow 5,\ 5\rightarrow 4,\  4\rightarrow 5\}
\end{equation}
In response to such algorithm seller $B$ cannot get more that $\pi(8,7)$ immediately and if seller $B$ believes that seller $A$ will keep this algorithm they do not have any incentives to continue with any other price pair. While for $A$ it is optimal to keep this algorithm because it guarantees higher than the monopoly expected payoff.

\begin{table}
\centering{}
\begin{tabular}{|c|c|c|c|c|c|}
\hline
$A\backslash B$ & $8$ & $7$ & $6$ & $5$ & $4$\tabularnewline
\hline 
$8$ & $2.95,\ 2.95$ & $2.41,\ 3.39$ & $1.86,\ 3.61$ & $1.36,\ 3.54$ & $0.95,\ 3.19$\tabularnewline
\hline 
$7$ & $3.39,\ 2.41$ & $2.87,\ 2.87$ & $2.29,\ 3.16$ & $1.74,\ 3.21$ & $1.25,\ 2.97$\tabularnewline
\hline 
$6$ & $3.61,\ 1.86$ & $3.16,\ 2.29$ & $2.64,\ 2.64$ & $2.08,\ 2.79$ & $1.55,\ 2.68$\tabularnewline
\hline 
$5$ & $3.54,\ 1.36$ & $3.21,\ 1.74$ & $2.79,\ 2.08$ & $2.30,\ 2.30$ & $1.80,\ 2.32$\tabularnewline
\hline 
$4$ & $3.19,\ 0.95$ & $2.97,\ 1.25$ & $2.68,\ 1.55$ & $2.32,\ 1.80$ & $1.90,\ 1.90$\tabularnewline
\hline
\end{tabular}
\caption{\label{tab:UE_payoffs} Expected payoffs from one customer in an unequal equilibrium.}
\end{table}

\begin{table}
\centering{}
$\qquad$%

\begin{tabular}{cc}
Seller $A$ & seller $B$  \\
\begin{tabular}{|c|c|c|c|c|c|}
\hline 
$A\backslash B$ & $8$ & $7$ & $6$ & $5$ & $4$\tabularnewline
\hline 
$8$ & $8,\ 8$ & $8,\ 7$ & $8,\ 6$ & $8,\ 6$ & $8,\ 7$\tabularnewline
\hline 
$7$ & $7,\ 8$ & $8,\ 8$ & $8,\ 7$ & $8,\ 7$ & $8,\ 7$\tabularnewline
\hline 
$6$ & $5,\ 4$ & $7,\ 8$ & $8,\ 8$ & $8,\ 8$ & $8,\ 8$\tabularnewline
\hline 
$5$ & $5,\ 4$ & $7,\ 8$ & $7,\ 8$ & $8,\ 8$ & $7,\ 8$\tabularnewline
\hline 
$4$ & $5,\ 4$ & $7,\ 8$ & $7,\ 8$ & $7,\ 8$ & $8,\ 8$\tabularnewline
\hline 
\end{tabular} &
\centering{}
\begin{tabular}{|c|c|c|c|c|c|}
\hline 
$A\backslash B$ & $8$ & $7$ & $6$ & $5$ & $4$\tabularnewline
\hline 
$8$ & $8,\ 8$ & $7,\ 7$ & $7,\ 7$ & $7,\ 7$ & $7,\ 7$\tabularnewline
\hline 
$7$ & $7,\ 8$ & $8,\ 8$ & $8,\ 8$ & $8,\ 8$ & $8,\ 8$\tabularnewline
\hline 
$6$ & $7,\ 8$ & $7,\ 8$ & $8,\ 8$ & $8,\ 8$ & $8,\ 8$\tabularnewline
\hline 
$5$ & $7,\ 8$ & $7,\ 8$ & $7,\ 8$ & $8,\ 8$ & $7,\ 8$\tabularnewline
\hline 
$4$ & $7,\ 8$ & $7,\ 8$ & $7,\ 8$ & $7,\ 8$ & $8,\ 8$\tabularnewline
\hline 
\end{tabular}
\end{tabular}
\caption{\label{tab:UE_transition} Price transition matrices (the first price in each pair corresponds to seller $A$)}
\end{table}

\medskip

\subsection*{\it \underline{Part E}: Modeling experimentation}

In this section we discuss in more detail how sellers learn the opponent's algorithm. In particular, we argue that within our model, experimentation is relatively inexpensive (and converges to $0$ as $\Delta t\rightarrow 0$), and, thus, the sellers will always choose to discover the opponent's algorithm (at least enough to make a decision). As a result, for $\Delta t$ small enough, we can assume the seller to be best responding the opponent's algorithm.  

Whenever a seller gets an opportunity to revise their algorithm, they can also spend a short period of time experimenting with the opponent's algorithm (by manually choosing an initial short sequence of prices) before committing to a new algorithm themselves.

Suppose $U^{observable}$ is the highest payoff the seller currently adjusting their algorithm could get if the opponent's algorithm is directly observable. Now, the considered seller can experiment for the first $K$ periods (for $K$ different price levels) to fully recover the opponent's algorithm. Then, their ultimate payoff is then at least:
\begin{equation}
\begin{array}{lll}
U & \ge &\sum_{j=0}^{K-1}e^{-r\cdot j\Delta t}\cdot\frac{\lambda}{r}(1-e^{-r\Delta t})\cdot\pi(\dots,\dots)+e^{-r\cdot K\Delta t}\cdot U^{observable} \\
& \ge & U^{observable}-\frac{\lambda}{r}\left(1-e^{-r\cdot K\Delta t}\right)\Delta\pi_{max}
\end{array}
\end{equation}
where $\Delta\pi_{max}=max_{p,q,p',q'}\left|\pi(p,q)-\pi(p',q')\right|$.\footnote{In the expression we also ignore the probability simultaneous algorithm adjustments, which for small $\Delta t$ is a small probability event, and, thus, similarly would not materially affect the ultimate payoff.} The second term in the above expressions is converging to $0$ as $\Delta t\rightarrow0$, thus, the actual expected payoffs approximate $U^{observable}$ for small $\Delta t$.\footnote{Since the experimentation phase is not longer than $K$ prices adjustment cycles, it could be even shorter the period of convergence to a fixed price pair, which we already ignore in our model by Assumption 1.} While if the seller does not experiment and chooses the response immediately, their expected payoff cannot be higher that $U^{observable}$ and could be sufficiently lower (even for small $\Delta t$) because the chosen response might not be optimal for the specific opponent's algorithm. 

We should point out two aspects about more general approaches to modeling experimentation. First that algorithms could in principle experiment at arbitrary times to learn more about the opponent's algorithm. At a high level this would help our conceptual arguments because it would make learning even easier. However, it would make the theoretical analysis intractable by making it harder to separate the pure delegation (to the algorithm) phase from the experimentation phase. Our modeling of the experimentation phase immediately after the Poisson arrival makes the analysis tractable while keeping the basic (realistic) forces in tact. 

Second, for general algorithms with many possible internal states, experimentation could certainly be more costly and it might take more time. It might not even be possible to fully decipher the opponent's algorithm and you might need to choose your algorithm with some remaining uncertainly. In addition, the opponent's algorithm might have some irreversible state transitions (e.g. some sort of grim-trigger algorithm), thus, it might not be obvious sometimes whether you want to experiment at all.

It would, of course, be interesting, and also challenging, to study these more general approaches to experimentation in algorithmic pricing. It is certainly beyond the scope of this paper, and we leave it for future research.
\newpage

\bibliographystyle{abbrvnat}
\bibliography{PriceWithAlgo}

\begin{thebibliography}{29}
\providecommand{\natexlab}[1]{#1}
\providecommand{\url}[1]{\texttt{#1}}
\expandafter\ifx\csname urlstyle\endcsname\relax
  \providecommand{\doi}[1]{doi: #1}\else
  \providecommand{\doi}{doi: \begingroup \urlstyle{rm}\Url}\fi

\bibitem[Abreu and Rubinstein(1988)]{finite_automaton}
D.~Abreu and A.~Rubinstein.
\newblock The structure of {Nash Equilibrium} in repeated games with finite
  automata.
\newblock \emph{Econometrica}, 56\penalty0 (6):\penalty0 1259--1281, 1988.

\bibitem[Abreu et~al.(1990)Abreu, Pearce, and Stacchetti]{aps}
D.~Abreu, D.~Pearce, and E.~Stacchetti.
\newblock Toward a theory of discounted repeated games with imperfect
  monitoring.
\newblock \emph{Econometrica}, 58\penalty0 (5):\penalty0 1041--1063, 1990.

\bibitem[Asker(2010)]{asker_collusion}
J.~Asker.
\newblock A study of the internal organization of a bidding cartel.
\newblock \emph{American Economic Review}, 100\penalty0 (3):\penalty0 724--762,
  2010.

\bibitem[Asker et~al.(2021)Asker, Fershtman, and Pakes]{algo_pakes}
J.~Asker, C.~Fershtman, and A.~Pakes.
\newblock Artificial intelligence and pricing: The impact of algorithm design.
\newblock UCLA, Tel Aviv and Harvard, 2021.

\bibitem[Assada et~al.(2021)Assada, Clark, Ershovc, and
  Xud]{German_gasoline_AI}
S.~Assada, R.~Clark, D.~Ershovc, and L.~Xud.
\newblock Algorithmic pricing and competition: Empirical evidence from the
  {German} retail gasoline market.
\newblock CESifo Working Paper No. 8521, 2021.

\bibitem[Aumann(1959)]{aumann_folk}
R.~J. Aumann.
\newblock Acceptable points in general cooperative n-person games.
\newblock In A.~W. Tucker and R.~D. Luce, editors, \emph{Contributions to the
  Theory of Games IV, Annals of Mathematics Study 40}. Princeton University
  Press, 1959.

\bibitem[Awaya and Krishna(2016)]{awaya_krishna}
Y.~Awaya and V.~Krishna.
\newblock On communication and collusion.
\newblock \emph{American Economic Review}, 106\penalty0 (2):\penalty0 285--315,
  2016.

\bibitem[Banchio and Skrzypacz(2022)]{AI-auctions}
M.~Banchio and A.~Skrzypacz.
\newblock Artificial intelligence and auction design.
\newblock Stanford University, 2022.

\bibitem[Brown and MacKay(2021)]{competiton_algo}
Z.~Y. Brown and A.~MacKay.
\newblock Competition in pricing algorithms.
\newblock {\it American Economic Journal: Microeconomics}, forthcoming, 2021.

\bibitem[Calvano et~al.(2020{\natexlab{a}})Calvano, Calzolari, Denicolò,
  Harrington~Jr., and Pastorello]{science_AI_collusion}
E.~Calvano, G.~Calzolari, V.~Denicolò, J.~E. Harrington~Jr., and
  S.~Pastorello.
\newblock Protecting consumers from collusive prices due to ai.
\newblock \emph{Science}, 370\penalty0 (6520):\penalty0 1040--1042,
  2020{\natexlab{a}}.

\bibitem[Calvano et~al.(2020{\natexlab{b}})Calvano, Calzolari, Denicolò, and
  Pastorello]{experiment_AI_collusion}
E.~Calvano, G.~Calzolari, V.~Denicolò, and S.~Pastorello.
\newblock Artificial intelligence, algorithmic pricing, and collusion.
\newblock \emph{American Economic Review}, 110\penalty0 (10):\penalty0
  3267--3297, 2020{\natexlab{b}}.

\bibitem[Chassang et~al.(2022)Chassang, Kawai, Nakabayashi, and
  Ortner]{chassang_collusion}
S.~Chassang, K.~Kawai, J.~Nakabayashi, and J.~Ortner.
\newblock Robust screens for noncompetitive bidding in procurement auctions.
\newblock \emph{Econometrica}, 90\penalty0 (1):\penalty0 315--346, 2022.

\bibitem[Friedman(1971)]{old_coop}
J.~W. Friedman.
\newblock A non-cooperative equilibrium for supergames.
\newblock \emph{The Review of Economic Studies}, 38\penalty0 (1):\penalty0
  1--12, 1971.

\bibitem[Green and Porter(1984)]{green_porter}
E.~J. Green and R.~H. Porter.
\newblock Noncooperative collusion under imperfect price information.
\newblock \emph{Econometrica}, 52\penalty0 (1):\penalty0 87--100, 1984.

\bibitem[Hansen et~al.(2021)Hansen, Misra, and Pai]{mallesh_bhai_ka_paper}
K.~T. Hansen, K.~Misra, and M.~M. Pai.
\newblock Algorithmic collusion: Supra-competitive prices via independent
  algorithms.
\newblock \emph{Marketing Science}, 40\penalty0 (1):\penalty0 1--12, 2021.

\bibitem[Harrington~Jr.(2018)]{harrington_ai}
J.~E. Harrington~Jr.
\newblock Developing competition law for collusion by autonomous agents.
\newblock \emph{Journal of Competition Law \& Economics}, 14\penalty0
  (3):\penalty0 331--363, 2018.

\bibitem[Harrington~Jr.(2022)]{harrington_algo}
J.~E. Harrington~Jr.
\newblock The effect of outsourcing pricing algorithms on market competition.
\newblock Management Science, {\it forthcoming}, 2022.

\bibitem[Johnson et~al.(2021)Johnson, Rhodes, and Wildenbeest]{platform_algo}
J.~Johnson, A.~Rhodes, and M.~Wildenbeest.
\newblock Platform design when sellers use pricing algorithms.
\newblock Cornell University and Toulouse School of Economics and Indiana
  University, 2021.

\bibitem[Kamada and Kandori(2020)]{revision_games}
Y.~Kamada and M.~Kandori.
\newblock Revision games.
\newblock \emph{Econometrica}, 88\penalty0 (4):\penalty0 1599--1630, 2020.

\bibitem[Kumar(2012)]{arthsshastra_collusion}
V.~Kumar.
\newblock Cartels in the kautiliya arthasastra.
\newblock \emph{Czech Economic Review}, 6\penalty0 (1):\penalty0 59--79, 2012.

\bibitem[Lagunoff and Matsui(1997)]{asyn_rg}
R.~Lagunoff and A.~Matsui.
\newblock Asynchronous choice in repeated coordination games.
\newblock \emph{Econometrica}, 65\penalty0 (6):\penalty0 1467--1477, 1997.

\bibitem[Marshall and Marx(2012)]{mm_collusion_book}
R.~C. Marshall and L.~M. Marx.
\newblock \emph{The Economics of Collusion: Cartels and Bidding Rings}.
\newblock MIT Press, 2012.

\bibitem[Maskin and Tirole(1988{\natexlab{a}})]{maskin_tirole1}
E.~Maskin and J.~Tirole.
\newblock A theory of dynamic oligopoly, {I}: Overview and quantity competition
  with large fixed costs.
\newblock \emph{Econometrica}, 56\penalty0 (3):\penalty0 549--569,
  1988{\natexlab{a}}.

\bibitem[Maskin and Tirole(1988{\natexlab{b}})]{maskin_tirole2}
E.~Maskin and J.~Tirole.
\newblock A theory of dynamic oligopoly, {II}: Price competition, kinked demand
  curves, and edgeworth cycles.
\newblock \emph{Econometrica}, 56\penalty0 (3):\penalty0 571--599,
  1988{\natexlab{b}}.

\bibitem[Miller and Watson(2013)]{miller_watson_renego}
D.~A. Miller and J.~Watson.
\newblock A theory of disagreement in repeated games with bargaining.
\newblock \emph{Econometrica}, 81\penalty0 (6):\penalty0 2303--2350, 2013.

\bibitem[Pearce(1987)]{david_renego}
D.~G. Pearce.
\newblock Renegotiation-proof equilibria: Collective rationality and
  intertemporal cooperation.
\newblock {\it Cowles Foundation Discussion Paper No. 855}, 1987.

\bibitem[Safronov and Strulovici(2018)]{contest_norms}
M.~Safronov and B.~Strulovici.
\newblock Contestable norms.
\newblock University of Cambridge and Northwestern University, 2018.

\bibitem[Salcedo(2015)]{bruno_tacit}
B.~Salcedo.
\newblock Pricing algorithms and tacit collusion.
\newblock University of Western Ontario, 2015.

\bibitem[Sugaya(2020)]{sugaya_folk}
T.~Sugaya.
\newblock Folk theorem in repeated games with private monitoring.
\newblock Review of Economic Studies, {\it forthcoming}, 2020.

\end{thebibliography}

\end{document}